\documentclass[lettersize,journal]{IEEEtran}
\usepackage{xcolor,soul,framed,caption} 

\usepackage[dvipsnames,svgnames,x11names]{xcolor}

\usepackage[pdftex]{graphicx}
\graphicspath{{../pdf/}{../jpeg/}}
\DeclareGraphicsExtensions{.pdf,.jpeg,.png}

\usepackage{enumitem}

\usepackage[cmex10]{amsmath}

\usepackage{array}
\usepackage{mdwmath}
\usepackage{mdwtab}
\usepackage{eqparbox}
\usepackage{url}
\usepackage{supertabular}                        \usepackage[colorlinks=true,
    linkcolor=blue,
    urlcolor=blue,
    citecolor=blue]{hyperref}    
\usepackage{amsmath,amssymb}

\usepackage{amsthm}
\newtheorem{proposition}{Proposition}

\usepackage{caption}
\usepackage{subcaption}
\usepackage{float}
\usepackage{booktabs}
\usepackage{multirow}
\usepackage{graphicx}
\usepackage{color}
\usepackage{tabularray}
\usepackage{hhline}
\usepackage{wasysym}
\usepackage{algorithm}
\usepackage{algorithmic}
\usepackage{tikz}
\usetikzlibrary{positioning,fit,shapes.misc,arrows.meta,calc}

\usepackage{tikz}
\usetikzlibrary{shapes,arrows.meta,positioning,fit,calc}

\usepackage{lipsum}

\usepackage{setspace}

\usepackage{titlesec}
\titlespacing*{\section}{0pt}{0.5\baselineskip}{0.3\baselineskip}
\titlespacing*{\subsection}{0pt}{0.4\baselineskip}{0.2\baselineskip}
\titlespacing*{\subsubsection}{0pt}{0.3\baselineskip}{0.1\baselineskip}

\makeatletter
\def\@IEEENORMtitlevspace{1\baselineskip}
\def\@IEEEMINtitlevspace{0.3\baselineskip}  
\makeatother


\begin{document}

\bstctlcite{IEEEexample:BSTcontrol}

\title{Constrained Federated Reinforcement Learning for Secure Offloading in Vehicular Edge Computing}

\author{%
  Zhiyu~Wang,~Muhammed~F.~Esgin,~and~Mohammad~Goudarzi%
  \textsuperscript{*}\\[0.5em]
  {\normalsize
    Faculty of Information Technology, Monash University, Australia%
  }%
  \thanks{%
    \textsuperscript{*}Corresponding author\newline
    Emails: zhiyu.wang@monash.edu;
    muhammed.esgin@monash.edu;
    mohammad.goudarzi@monash.edu.
  }%
}

\maketitle
\IEEEaftertitletext{\vspace{-1.5em}}

\begin{abstract}
Vehicular Edge Computing (VEC) reduces latency for computation-intensive vehicular applications by offloading tasks to roadside edge servers. Real vehicular workloads, however, consist of dependent subtasks with stringent reliability needs, while the cross-region collaborative learning used to adapt offloading policies must stay robust to heterogeneous data and secure over open wireless channels. We propose CRAFT, a Constrained Reinforcement-based Adaptive Federated Task Offloading framework that casts Directed Acyclic Graph (DAG)-structured offloading and resource allocation as a Constrained Markov Decision Process. CRAFT couples centralized-critic policy learning with Lagrangian relaxation to explicitly regulate the long-term task-completion constraint, and uses two-tier federated aggregation with proximal regularization to jointly mitigate non-Independent and Identically Distributed (non-IID) drift and stabilize constraint adaptation. Confidentiality and integrity for federated parameter exchange are provided by post-quantum homomorphic encryption and lattice-based signatures with modest amortized overhead. On a Virtual Machine-based testbed over a realistic Melbourne Central Business District network with real computation kernels, CRAFT cuts weighted cost by up to 42\% over state-of-the-art baselines, maintains a completion rate of around 96\%, and adds 6.4\% security overhead.
\end{abstract}


\IEEEpeerreviewmaketitle

\section{Introduction}
\label{sec:introduction}

The rapid advancement of connected and autonomous vehicles has given rise to computation-intensive applications such as autonomous driving perception, augmented reality navigation, and in-vehicle voice assistants~\cite{gong2023edge_its, katare2023edge_ad}, all imposing stringent real-time computing requirements. Constrained by the limited computing resources and battery capacity of onboard platforms, relying solely on local execution can hardly satisfy these demands. Vehicular Edge Computing (VEC), which offloads tasks to Edge Servers (ESs) deployed on Roadside Units (RSUs), offers a promising solution. However, VEC task offloading remains extremely challenging due to rapid channel variations caused by vehicle mobility, an exponentially growing joint decision space~\cite{nguyen2020deep}, and the need to balance conflicting metrics including latency, energy consumption, and task completion rate.

Beyond these general challenges, four deeper technical difficulties persist. First, real vehicular applications such as autonomous driving perception pipelines inherently comprise multiple Directed Acyclic Graph (DAG)-structured processing stages with data-flow precedence constraints, requiring subtask-level parallel scheduling along the critical path~\cite{chen2024e2ead}. Second, task completion rate is a safety-critical metric that demands formal long-term guarantees~\cite{3gpp22886, iso26262}. Third, vehicles independently trained within isolated edge regions inevitably converge to region-specific strategies, causing severe policy mismatch when vehicles cross coverage boundaries. Federated Learning (FL) offers a natural solution for cross-region knowledge sharing~\cite{zhu2024flight,liu-ccgrid26}, but the significant regional differences in traffic density, task composition, and vehicle populations give rise to strongly heterogeneous data distributions that severely compromise aggregation quality and learning stability~\cite{lim2020federated}. Fourth, the cross-region federated collaboration introduced for the third challenge is itself the source of a new security problem: parameters and training signals must be repeatedly exchanged over open Vehicle-to-Infrastructure (V2I) channels that are susceptible to eavesdropping and tampering~\cite{yoshizawa2023survey}, risking confidentiality leakage of vehicles' behavioral patterns (not only to external eavesdroppers but also to the honest-but-curious edge servers that process the parameters), as well as integrity violations that propagate through aggregation to all participants. Moreover, as long-lifecycle assets, vehicles face harvest-now-decrypt-later threats~\cite{bhargavan2024formal}, and with major providers including Google\footnote{\url{https://cloud.google.com/security/resources/post-quantum-cryptography}} and Cloudflare\footnote{\url{https://blog.cloudflare.com/post-quantum-roadmap/}} migrating to Post-Quantum Cryptography (PQC) by 2029, post-quantum resilience is imperative from the outset.

Addressing these challenges calls for an offloading policy that adapts online to dynamic mobility and channel conditions without relying on accurate system models. Deep Reinforcement Learning (DRL), which learns such policies directly through environment interaction, is naturally suited to this setting \cite{wang2025tf}, while a multi-agent formulation is needed to avoid the exponentially growing joint action space that single-agent DRL incurs in multi-vehicle multi-server scenarios \cite{nguyen2020deep}.

To this end, we propose CRAFT (Constrained Reinforcement-based Adaptive Federated Task offloading), a secure federated Multi-Agent Reinforcement Learning (MARL) framework that models DAG-structured offloading in VEC as a Constrained Markov Decision Process (CMDP) and addresses the four challenges through four cooperating components: (i) a Centralized Training with Decentralized Execution (CTDE)-based multi-agent module with Actors on vehicles for offloading decisions and Critics on ESs for value estimation; (ii) a Lagrangian relaxation module that folds the long-term completion-rate constraint into the reward via online dual updates; (iii) a two-tier federated Actor aggregation module whose proximal regularization simultaneously curbs non-Independent and Identically Distributed (non-IID) drift and damps the policy oscillations from Lagrangian updates; and (iv) a post-quantum security module combining Cheon-Kim-Kim-Song (CKKS)~\cite{cheon2017ckks} Homomorphic Encryption (HE) for confidentiality with Module-Lattice Digital Signature Algorithm (ML-DSA) (a.k.a. Dilithium)~\cite{ducas2018crystals, fips204} signatures for integrity over open V2I channels.

The main contributions are summarized as follows:
\begin{itemize}
    \item We formulate DAG-structured VEC offloading capturing subtask structure, dependency, and heterogeneity, and model it as a CMDP with a long-term task completion rate constraint.
    \item We design a CTDE-based multi-agent constrained policy optimization scheme with Lagrangian relaxation and projected dual updates, and prove that the resulting long-term constraint violation decreases at an $O(1/\sqrt{J})$ rate up to a residual term induced by inexact neural policy optimization, where $J$ is the total number of Lagrangian multiplier-update rounds.
    \item We propose a two-tier federated aggregation with a dual-role proximal regularization mechanism that jointly handles non-IID distributions and stabilizes constrained optimization, with confidentiality and integrity secured by post-quantum CKKS and ML-DSA, respectively.
    \item We build a Virtual Machine (VM)-based testbed with real vehicular workloads; on a realistic Melbourne Central Business District (CBD) network, CRAFT outperforms three state-of-the-art baselines by about 16\%, 20\%, and 42\% in weighted cost while holding a 96\% task completion rate near the constraint threshold.
\end{itemize}

The remainder of this paper covers related work and research gaps (Section~\ref{sec:related_work}), the system model and problem formulation (Section~\ref{sec:system_model}), the CRAFT framework with its convergence/stability and complexity analysis (Section~\ref{sec:algorithm}), experimental results (Section~\ref{sec:experiments}), and conclusions (Section~\ref{sec:conclusion}).

\section{Related Work}
\label{sec:related_work}

\subsection{Task Offloading in Vehicular Edge Computing}

Model-based optimization methods derive theoretically optimal offloading solutions through rigorous mathematical formulations, predominantly via game theory. Sun~\emph{et~al.}~\cite{10024868} combined a Bargaining Game for intra-server allocation with Matching-based inter-server offloading to maximize utility, Fan~\emph{et~al.}~\cite{10035469} adopted a Potential Game to minimize processing delay, while Yin~\emph{et~al.}~\cite{10414410} modeled the system as a multi-stage Stackelberg Game to maximize tripartite benefits. Several works further couple optimization with security: Xiao~\emph{et~al.}~\cite{xiao2023joint} integrated block-cipher encryption into a Generalized Benders Decomposition (GBD) scheme, Yang~\emph{et~al.}~\cite{10745593} secured a pricing-based Stackelberg Game for platoon offloading via Advanced Encryption Standard (AES), and Liang~\emph{et~al.}~\cite{10969549} adopted a blockchain framework with game-theoretic pricing. However, these methods rely on accurate model assumptions and adapt poorly to rapid mobility and time-varying channels~\cite{zabihi2023reinforcement, wang2024deep}.

Reinforcement Learning (RL) enables online adaptive decision-making without accurate system models. Qiu~\emph{et~al.}~\cite{10517901}, Fan~\emph{et~al.}~\cite{10400987}, and Ju~\emph{et~al.}~\cite{ju2024eeso} respectively applied Deep Deterministic Policy Gradient (DDPG), Twin Delayed Deep Deterministic Policy Gradient (TD3), and Asynchronous Advantage Actor-Critic (A3C) to joint offloading and resource/power allocation, with the last leveraging Physical Layer Security (PLS) against eavesdroppers. However, such Single-Agent RL (SARL) methods suffer from exponentially growing joint action spaces in multi-vehicle multi-server settings, severely limiting scalability~\cite{nguyen2020deep}.

Multi-Agent RL (MARL) mitigates this by assigning each vehicle a learning agent. Abbas~\emph{et~al.}~\cite{10779474} proposed DPTORA based on the Generalist Pursuit Learning Algorithm (GPLA) for flexible partial offloading across RSUs. Several Multi-Agent DDPG (MADDPG)-based schemes target different objectives: Xu~\emph{et~al.}~\cite{xu2023joint} maximized task completion rate, Fofana~\emph{et~al.}~\cite{10598376} balanced time, cost, and success probability, and Xiao~\emph{et~al.}~\cite{11536917} embedded PLS into the reward for secure offloading. Wei~\emph{et~al.}~\cite{10056271} adopted Multi-Agent Soft Actor-Critic (MASAC) to optimise latency, while Ju~\emph{et~al.}~\cite{ju2023sora} used a PLS-aware Multi-Agent Double Deep Q-Network (MADDQN) to minimize processing delay against eavesdroppers. More recently, FL has been integrated for collaborative knowledge sharing: Wu~\emph{et~al.}~\cite{10720029} studied auxiliary vehicle resource exploitation, Lei~\emph{et~al.}~\cite{11085134} proposed Fed-MADDPG for joint scheduling without state sharing, and Zeng~\emph{et~al.}~\cite{11096732} proposed a Multi-Agent Proximal Policy Optimization (MAPPO)-based method minimizing average delay.

\subsection{A Qualitative Comparison}
In what follows, we systematically identify existing works' limitations and position the CRAFT contributions.

\subsubsection{Comparative Analysis Dimensions}
Table~\ref{tab:related_works} compares CRAFT with existing works across several dimensions. \emph{Task Properties} captures workload complexity via \emph{Subtasks} (single unit vs.\ multiple subtasks), \emph{Dependency} (inter-subtask execution order), and \emph{Heterogeneity} (diverse computational characteristics). \emph{Methodology} covers \emph{Paradigm} (optimization, single-agent, or multi-agent RL), \emph{Algorithm} (the concrete technique), \emph{Constraint} (a formal mechanism for long-term guarantees), and \emph{FL} (federated learning for cross-region sharing). \emph{Non-IID} marks explicit handling of heterogeneous cross-region data. \emph{Security} spans \emph{Confidentiality}, \emph{Integrity} (cryptographic integrity/authenticity verification), and \emph{PQC} (use of PQC to protect the exchanged/aggregated parameters). Confidentiality is further split by adversary into \emph{vs.\ Eaves.} (network eavesdroppers) and \emph{vs.\ ES} (the honest-but-curious edge server that can access the plaintext). \emph{Optimization Objectives} lists targeted metrics (latency, energy, completion rate, multi-objective). \emph{Evaluation} distinguishes pure simulation from real-workload testbeds.

\renewcommand{\arraystretch}{1.1}
\begin{table*}[t]
\centering
\caption{Qualitative comparison of CRAFT with existing related works.}
\label{tab:related_works}
\resizebox{\textwidth}{!}{%
\begin{tabular}{cccccccccccccccccc}
\hline
\multicolumn{1}{|c|}{\multirow{3}{*}{\textbf{Work}}} &
\multicolumn{3}{c|}{\textbf{Task Properties}} &
\multicolumn{4}{c|}{\textbf{Methodology}} &
\multicolumn{1}{c|}{\multirow{3}{*}{\textbf{Non-IID}}} &
\multicolumn{4}{c|}{\textbf{Security}} &
\multicolumn{4}{c|}{\textbf{Optimization Objectives}} &
\multicolumn{1}{c|}{\multirow{3}{*}{\textbf{Evaluation}}} \\
\cline{2-4} \cline{5-8} \cline{10-13} \cline{14-17}
\multicolumn{1}{|c|}{} &
\multicolumn{1}{c|}{\multirow{2}{*}{\textbf{Subtasks}}} &
\multicolumn{1}{c|}{\multirow{2}{*}{\textbf{Dependency}}} &
\multicolumn{1}{c|}{\multirow{2}{*}{\textbf{Heterogeneity}}} &
\multicolumn{2}{c|}{\textbf{Main Technique}} &
\multicolumn{1}{c|}{\multirow{2}{*}{\textbf{Constraint}}} &
\multicolumn{1}{c|}{\multirow{2}{*}{\textbf{FL}}} &
\multicolumn{1}{c|}{} &
\multicolumn{2}{c|}{\textbf{Confidentiality}} &
\multicolumn{1}{c|}{\multirow{2}{*}{\textbf{Integrity}}} &
\multicolumn{1}{c|}{\multirow{2}{*}{\textbf{PQC}}} &
\multicolumn{1}{c|}{\multirow{2}{*}{\textbf{Latency}}} &
\multicolumn{1}{c|}{\multirow{2}{*}{\textbf{Energy}}} &
\multicolumn{1}{c|}{\multirow{2}{*}{\begin{tabular}[c]{@{}c@{}}\textbf{Compl.}\\\textbf{Rate}\end{tabular}}} &
\multicolumn{1}{c|}{\multirow{2}{*}{\begin{tabular}[c]{@{}c@{}}\textbf{Multi}\\\textbf{Obj}\end{tabular}}} &
\multicolumn{1}{c|}{} \\
\cline{5-6}\cline{10-11}
\multicolumn{1}{|c|}{} &
\multicolumn{1}{c|}{} &
\multicolumn{1}{c|}{} &
\multicolumn{1}{c|}{} &
\multicolumn{1}{c|}{\textbf{Paradigm}} &
\multicolumn{1}{c|}{\textbf{Algorithm}} &
\multicolumn{1}{c|}{} &
\multicolumn{1}{c|}{} &
\multicolumn{1}{c|}{} &
\multicolumn{1}{c|}{\textbf{vs. Eaves.}} &
\multicolumn{1}{c|}{\textbf{vs. ES}} &
\multicolumn{1}{c|}{} &
\multicolumn{1}{c|}{} &
\multicolumn{1}{c|}{} &
\multicolumn{1}{c|}{} &
\multicolumn{1}{c|}{} &
\multicolumn{1}{c|}{} &
\multicolumn{1}{c|}{} \\
\hline
\multicolumn{1}{|c|}{Sun~\emph{et~al.}~\cite{10024868}} &
\multicolumn{1}{c|}{Single} &
\multicolumn{1}{c|}{$\times$} &
\multicolumn{1}{c|}{$\checkmark$} &
\multicolumn{1}{c|}{\multirow{6}{*}{Optimization}} &
\multicolumn{1}{c|}{\begin{tabular}[c]{@{}c@{}}Bargaining Game\\+ Matching\end{tabular}} &
\multicolumn{1}{c|}{$\times$} &
\multicolumn{1}{c|}{$\times$} &
\multicolumn{1}{c|}{$\times$} &
\multicolumn{1}{c|}{$\times$} &
\multicolumn{1}{c|}{$\times$} &
\multicolumn{1}{c|}{$\times$} &
\multicolumn{1}{c|}{$\times$} &
\multicolumn{1}{c|}{$\checkmark$} &
\multicolumn{1}{c|}{$\checkmark$} &
\multicolumn{1}{c|}{$\times$} &
\multicolumn{1}{c|}{$\checkmark$} &
\multicolumn{1}{c|}{Simulation} \\
\cline{1-4}\cline{6-18}
\multicolumn{1}{|c|}{Fan~\emph{et~al.}~\cite{10035469}} &
\multicolumn{1}{c|}{Single} &
\multicolumn{1}{c|}{$\times$} &
\multicolumn{1}{c|}{$\times$} &
\multicolumn{1}{c|}{} &
\multicolumn{1}{c|}{Potential Game} &
\multicolumn{1}{c|}{$\times$} &
\multicolumn{1}{c|}{$\times$} &
\multicolumn{1}{c|}{$\times$} &
\multicolumn{1}{c|}{$\times$} &
\multicolumn{1}{c|}{$\times$} &
\multicolumn{1}{c|}{$\times$} &
\multicolumn{1}{c|}{$\times$} &
\multicolumn{1}{c|}{$\checkmark$} &
\multicolumn{1}{c|}{$\times$} &
\multicolumn{1}{c|}{$\times$} &
\multicolumn{1}{c|}{$\times$} &
\multicolumn{1}{c|}{Simulation} \\
\cline{1-4}\cline{6-18}
\multicolumn{1}{|c|}{Yin~\emph{et~al.}~\cite{10414410}} &
\multicolumn{1}{c|}{Single} &
\multicolumn{1}{c|}{$\times$} &
\multicolumn{1}{c|}{$\times$} &
\multicolumn{1}{c|}{} &
\multicolumn{1}{c|}{Stackelberg Game} &
\multicolumn{1}{c|}{$\times$} &
\multicolumn{1}{c|}{$\times$} &
\multicolumn{1}{c|}{$\times$} &
\multicolumn{1}{c|}{$\times$} &
\multicolumn{1}{c|}{$\times$} &
\multicolumn{1}{c|}{$\times$} &
\multicolumn{1}{c|}{$\times$} &
\multicolumn{1}{c|}{$\checkmark$} &
\multicolumn{1}{c|}{$\checkmark$} &
\multicolumn{1}{c|}{$\times$} &
\multicolumn{1}{c|}{$\checkmark$} &
\multicolumn{1}{c|}{Simulation} \\
\cline{1-4}\cline{6-18}
\multicolumn{1}{|c|}{Xiao~\emph{et~al.}~\cite{xiao2023joint}} &
\multicolumn{1}{c|}{Single} &
\multicolumn{1}{c|}{$\times$} &
\multicolumn{1}{c|}{$\checkmark$} &
\multicolumn{1}{c|}{} &
\multicolumn{1}{c|}{GBD} &
\multicolumn{1}{c|}{$\times$} &
\multicolumn{1}{c|}{$\times$} &
\multicolumn{1}{c|}{$\times$} &
\multicolumn{1}{c|}{Block Cipher} &
\multicolumn{1}{c|}{$\times$} &
\multicolumn{1}{c|}{$\times$} &
\multicolumn{1}{c|}{$\times$} &
\multicolumn{1}{c|}{$\checkmark$} &
\multicolumn{1}{c|}{$\checkmark$} &
\multicolumn{1}{c|}{$\times$} &
\multicolumn{1}{c|}{$\checkmark$} &
\multicolumn{1}{c|}{Simulation} \\
\cline{1-4}\cline{6-18}
\multicolumn{1}{|c|}{Yang~\emph{et~al.}~\cite{10745593}} &
\multicolumn{1}{c|}{Single} &
\multicolumn{1}{c|}{$\times$} &
\multicolumn{1}{c|}{$\times$} &
\multicolumn{1}{c|}{} &
\multicolumn{1}{c|}{Stackelberg Game} &
\multicolumn{1}{c|}{$\times$} &
\multicolumn{1}{c|}{$\times$} &
\multicolumn{1}{c|}{$\times$} &
\multicolumn{1}{c|}{AES} &
\multicolumn{1}{c|}{$\times$} &
\multicolumn{1}{c|}{$\times$} &
\multicolumn{1}{c|}{$\times$} &
\multicolumn{1}{c|}{$\checkmark$} &
\multicolumn{1}{c|}{$\checkmark$} &
\multicolumn{1}{c|}{$\times$} &
\multicolumn{1}{c|}{$\checkmark$} &
\multicolumn{1}{c|}{Simulation} \\
\cline{1-4}\cline{6-18}
\multicolumn{1}{|c|}{Liang~\emph{et~al.}~\cite{10969549}} &
\multicolumn{1}{c|}{Single} &
\multicolumn{1}{c|}{$\times$} &
\multicolumn{1}{c|}{$\times$} &
\multicolumn{1}{c|}{} &
\multicolumn{1}{c|}{Stackelberg Game} &
\multicolumn{1}{c|}{$\times$} &
\multicolumn{1}{c|}{$\times$} &
\multicolumn{1}{c|}{$\times$} &
\multicolumn{1}{c|}{$\times$} &
\multicolumn{1}{c|}{$\times$} &
\multicolumn{1}{c|}{Blockchain} &
\multicolumn{1}{c|}{$\times$} &
\multicolumn{1}{c|}{$\checkmark$} &
\multicolumn{1}{c|}{$\checkmark$} &
\multicolumn{1}{c|}{$\times$} &
\multicolumn{1}{c|}{$\checkmark$} &
\multicolumn{1}{c|}{Simulation} \\
\hline
\multicolumn{1}{|c|}{Qiu~\emph{et~al.}~\cite{10517901}} &
\multicolumn{1}{c|}{Single} &
\multicolumn{1}{c|}{$\times$} &
\multicolumn{1}{c|}{$\times$} &
\multicolumn{1}{c|}{\multirow{3}{*}{SARL}} &
\multicolumn{1}{c|}{DDPG} &
\multicolumn{1}{c|}{$\times$} &
\multicolumn{1}{c|}{$\times$} &
\multicolumn{1}{c|}{$\times$} &
\multicolumn{1}{c|}{$\times$} &
\multicolumn{1}{c|}{$\times$} &
\multicolumn{1}{c|}{$\times$} &
\multicolumn{1}{c|}{$\times$} &
\multicolumn{1}{c|}{$\checkmark$} &
\multicolumn{1}{c|}{$\checkmark$} &
\multicolumn{1}{c|}{$\times$} &
\multicolumn{1}{c|}{$\checkmark$} &
\multicolumn{1}{c|}{Simulation} \\
\cline{1-4}\cline{6-18}
\multicolumn{1}{|c|}{Fan~\emph{et~al.}~\cite{10400987}} &
\multicolumn{1}{c|}{Single} &
\multicolumn{1}{c|}{$\times$} &
\multicolumn{1}{c|}{$\times$} &
\multicolumn{1}{c|}{} &
\multicolumn{1}{c|}{TD3} &
\multicolumn{1}{c|}{$\times$} &
\multicolumn{1}{c|}{$\times$} &
\multicolumn{1}{c|}{$\times$} &
\multicolumn{1}{c|}{$\times$} &
\multicolumn{1}{c|}{$\times$} &
\multicolumn{1}{c|}{$\times$} &
\multicolumn{1}{c|}{$\times$} &
\multicolumn{1}{c|}{$\checkmark$} &
\multicolumn{1}{c|}{$\checkmark$} &
\multicolumn{1}{c|}{$\times$} &
\multicolumn{1}{c|}{$\checkmark$} &
\multicolumn{1}{c|}{Simulation} \\
\cline{1-4}\cline{6-18}
\multicolumn{1}{|c|}{Ju~\emph{et~al.}~\cite{ju2024eeso}} &
\multicolumn{1}{c|}{Single} &
\multicolumn{1}{c|}{$\times$} &
\multicolumn{1}{c|}{$\checkmark$} &
\multicolumn{1}{c|}{} &
\multicolumn{1}{c|}{A3C} &
\multicolumn{1}{c|}{$\times$} &
\multicolumn{1}{c|}{$\times$} &
\multicolumn{1}{c|}{$\times$} &
\multicolumn{1}{c|}{PLS} &
\multicolumn{1}{c|}{$\times$} &
\multicolumn{1}{c|}{$\times$} &
\multicolumn{1}{c|}{$\times$} &
\multicolumn{1}{c|}{$\checkmark$} &
\multicolumn{1}{c|}{$\checkmark$} &
\multicolumn{1}{c|}{$\times$} &
\multicolumn{1}{c|}{$\checkmark$} &
\multicolumn{1}{c|}{Simulation} \\
\hline
\multicolumn{1}{|c|}{Abbas~\emph{et~al.}~\cite{10779474}} &
\multicolumn{1}{c|}{Single} &
\multicolumn{1}{c|}{$\times$} &
\multicolumn{1}{c|}{$\checkmark$} &
\multicolumn{1}{c|}{\multirow{10}{*}{MARL}} &
\multicolumn{1}{c|}{GPLA} &
\multicolumn{1}{c|}{$\times$} &
\multicolumn{1}{c|}{$\times$} &
\multicolumn{1}{c|}{$\times$} &
\multicolumn{1}{c|}{$\times$} &
\multicolumn{1}{c|}{$\times$} &
\multicolumn{1}{c|}{$\times$} &
\multicolumn{1}{c|}{$\times$} &
\multicolumn{1}{c|}{$\checkmark$} &
\multicolumn{1}{c|}{$\checkmark$} &
\multicolumn{1}{c|}{$\times$} &
\multicolumn{1}{c|}{$\checkmark$} &
\multicolumn{1}{c|}{Simulation} \\
\cline{1-4}\cline{6-18}
\multicolumn{1}{|c|}{Xu~\emph{et~al.}~\cite{xu2023joint}} &
\multicolumn{1}{c|}{Single} &
\multicolumn{1}{c|}{$\times$} &
\multicolumn{1}{c|}{$\checkmark$} &
\multicolumn{1}{c|}{} &
\multicolumn{1}{c|}{MADDPG} &
\multicolumn{1}{c|}{$\times$} &
\multicolumn{1}{c|}{$\times$} &
\multicolumn{1}{c|}{$\times$} &
\multicolumn{1}{c|}{$\times$} &
\multicolumn{1}{c|}{$\times$} &
\multicolumn{1}{c|}{$\times$} &
\multicolumn{1}{c|}{$\times$} &
\multicolumn{1}{c|}{$\times$} &
\multicolumn{1}{c|}{$\times$} &
\multicolumn{1}{c|}{$\checkmark$} &
\multicolumn{1}{c|}{$\times$} &
\multicolumn{1}{c|}{Simulation} \\
\cline{1-4}\cline{6-18}
\multicolumn{1}{|c|}{Fofana~\emph{et~al.}~\cite{10598376}} &
\multicolumn{1}{c|}{Single} &
\multicolumn{1}{c|}{$\times$} &
\multicolumn{1}{c|}{$\times$} &
\multicolumn{1}{c|}{} &
\multicolumn{1}{c|}{MADDPG} &
\multicolumn{1}{c|}{$\times$} &
\multicolumn{1}{c|}{$\times$} &
\multicolumn{1}{c|}{$\times$} &
\multicolumn{1}{c|}{$\times$} &
\multicolumn{1}{c|}{$\times$} &
\multicolumn{1}{c|}{$\times$} &
\multicolumn{1}{c|}{$\times$} &
\multicolumn{1}{c|}{$\checkmark$} &
\multicolumn{1}{c|}{$\checkmark$} &
\multicolumn{1}{c|}{$\checkmark$} &
\multicolumn{1}{c|}{$\checkmark$} &
\multicolumn{1}{c|}{Simulation} \\
\cline{1-4}\cline{6-18}
\multicolumn{1}{|c|}{Xiao~\emph{et~al.}~\cite{11536917}} &
\multicolumn{1}{c|}{Single} &
\multicolumn{1}{c|}{$\times$} &
\multicolumn{1}{c|}{$\times$} &
\multicolumn{1}{c|}{} &
\multicolumn{1}{c|}{MADDPG} &
\multicolumn{1}{c|}{$\times$} &
\multicolumn{1}{c|}{$\times$} &
\multicolumn{1}{c|}{$\times$} &
\multicolumn{1}{c|}{PLS} &
\multicolumn{1}{c|}{$\times$} &
\multicolumn{1}{c|}{$\times$} &
\multicolumn{1}{c|}{$\times$} &
\multicolumn{1}{c|}{$\checkmark$} &
\multicolumn{1}{c|}{$\checkmark$} &
\multicolumn{1}{c|}{$\times$} &
\multicolumn{1}{c|}{$\checkmark$} &
\multicolumn{1}{c|}{Simulation} \\
\cline{1-4}\cline{6-18}
\multicolumn{1}{|c|}{Wei~\emph{et~al.}~\cite{10056271}} &
\multicolumn{1}{c|}{Single} &
\multicolumn{1}{c|}{$\times$} &
\multicolumn{1}{c|}{$\times$} &
\multicolumn{1}{c|}{} &
\multicolumn{1}{c|}{MASAC} &
\multicolumn{1}{c|}{$\times$} &
\multicolumn{1}{c|}{$\times$} &
\multicolumn{1}{c|}{$\times$} &
\multicolumn{1}{c|}{$\times$} &
\multicolumn{1}{c|}{$\times$} &
\multicolumn{1}{c|}{$\times$} &
\multicolumn{1}{c|}{$\times$} &
\multicolumn{1}{c|}{$\checkmark$} &
\multicolumn{1}{c|}{$\times$} &
\multicolumn{1}{c|}{$\times$} &
\multicolumn{1}{c|}{$\checkmark$} &
\multicolumn{1}{c|}{Simulation} \\
\cline{1-4}\cline{6-18}
\multicolumn{1}{|c|}{Ju~\emph{et~al.}~\cite{ju2023sora}} &
\multicolumn{1}{c|}{Single} &
\multicolumn{1}{c|}{$\times$} &
\multicolumn{1}{c|}{$\checkmark$} &
\multicolumn{1}{c|}{} &
\multicolumn{1}{c|}{MADDQN} &
\multicolumn{1}{c|}{$\times$} &
\multicolumn{1}{c|}{$\times$} &
\multicolumn{1}{c|}{$\times$} &
\multicolumn{1}{c|}{PLS} &
\multicolumn{1}{c|}{$\times$} &
\multicolumn{1}{c|}{$\times$} &
\multicolumn{1}{c|}{$\times$} &
\multicolumn{1}{c|}{$\checkmark$} &
\multicolumn{1}{c|}{$\times$} &
\multicolumn{1}{c|}{$\times$} &
\multicolumn{1}{c|}{$\times$} &
\multicolumn{1}{c|}{Simulation} \\
\cline{1-4}\cline{6-18}
\multicolumn{1}{|c|}{Wu~\emph{et~al.}~\cite{10720029}} &
\multicolumn{1}{c|}{Single} &
\multicolumn{1}{c|}{$\times$} &
\multicolumn{1}{c|}{$\times$} &
\multicolumn{1}{c|}{} &
\multicolumn{1}{c|}{MADDPG} &
\multicolumn{1}{c|}{$\times$} &
\multicolumn{1}{c|}{$\checkmark$} &
\multicolumn{1}{c|}{$\times$} &
\multicolumn{1}{c|}{$\times$} &
\multicolumn{1}{c|}{$\times$} &
\multicolumn{1}{c|}{$\times$} &
\multicolumn{1}{c|}{$\times$} &
\multicolumn{1}{c|}{$\checkmark$} &
\multicolumn{1}{c|}{$\checkmark$} &
\multicolumn{1}{c|}{$\checkmark$} &
\multicolumn{1}{c|}{$\checkmark$} &
\multicolumn{1}{c|}{Simulation} \\
\cline{1-4}\cline{6-18}
\multicolumn{1}{|c|}{Lei~\emph{et~al.}~\cite{11085134}} &
\multicolumn{1}{c|}{Single} &
\multicolumn{1}{c|}{$\times$} &
\multicolumn{1}{c|}{$\times$} &
\multicolumn{1}{c|}{} &
\multicolumn{1}{c|}{MADDPG} &
\multicolumn{1}{c|}{$\times$} &
\multicolumn{1}{c|}{$\checkmark$} &
\multicolumn{1}{c|}{$\times$} &
\multicolumn{1}{c|}{$\times$} &
\multicolumn{1}{c|}{$\times$} &
\multicolumn{1}{c|}{$\times$} &
\multicolumn{1}{c|}{$\times$} &
\multicolumn{1}{c|}{$\checkmark$} &
\multicolumn{1}{c|}{$\times$} &
\multicolumn{1}{c|}{$\checkmark$} &
\multicolumn{1}{c|}{$\checkmark$} &
\multicolumn{1}{c|}{Simulation} \\
\cline{1-4}\cline{6-18}
\multicolumn{1}{|c|}{Zeng~\emph{et~al.}~\cite{11096732}} &
\multicolumn{1}{c|}{Single} &
\multicolumn{1}{c|}{$\times$} &
\multicolumn{1}{c|}{$\times$} &
\multicolumn{1}{c|}{} &
\multicolumn{1}{c|}{MAPPO} &
\multicolumn{1}{c|}{$\times$} &
\multicolumn{1}{c|}{$\checkmark$} &
\multicolumn{1}{c|}{$\times$} &
\multicolumn{1}{c|}{$\times$} &
\multicolumn{1}{c|}{$\times$} &
\multicolumn{1}{c|}{$\times$} &
\multicolumn{1}{c|}{$\times$} &
\multicolumn{1}{c|}{$\checkmark$} &
\multicolumn{1}{c|}{$\times$} &
\multicolumn{1}{c|}{$\times$} &
\multicolumn{1}{c|}{$\times$} &
\multicolumn{1}{c|}{Simulation} \\
\cline{1-4}\cline{6-18}
\multicolumn{1}{|c|}{\textbf{CRAFT (Ours)}} &
\multicolumn{1}{c|}{\textbf{Multiple}} &
\multicolumn{1}{c|}{\textbf{$\checkmark$}} &
\multicolumn{1}{c|}{\textbf{$\checkmark$}} &
\multicolumn{1}{c|}{} &
\multicolumn{1}{c|}{\textbf{MA Actor-Critic}} &
\multicolumn{1}{c|}{\textbf{CMDP+Lagrangian}} &
\multicolumn{1}{c|}{$\checkmark$} &
\multicolumn{1}{c|}{\textbf{Prox. Reg.}} &
\multicolumn{1}{c|}{\textbf{CKKS}} &
\multicolumn{1}{c|}{\textbf{CKKS}} &
\multicolumn{1}{c|}{\textbf{ML-DSA}} &
\multicolumn{1}{c|}{\textbf{$\checkmark$}} &
\multicolumn{1}{c|}{$\checkmark$} &
\multicolumn{1}{c|}{$\checkmark$} &
\multicolumn{1}{c|}{$\checkmark$} &
\multicolumn{1}{c|}{$\checkmark$} &
\multicolumn{1}{c|}{\textbf{Testbed}} \\
\hline
\end{tabular}%
}
\par\vspace{1mm}
\begin{minipage}{\textwidth}
\raggedright\footnotesize\linespread{1.0}\selectfont
\textit{}SARL = Single-Agent RL; MARL = Multi-Agent RL; DDPG = Deep Deterministic Policy Gradient; TD3 = Twin Delayed DDPG; A3C = Asynchronous Advantage Actor-Critic; GPLA = Generalist Pursuit Learning Algorithm; GBD = Generalized Benders Decomposition; MADDPG = Multi-Agent DDPG; MASAC = Multi-Agent Soft Actor-Critic; MAPPO = Multi-Agent Proximal Policy Optimization; MADDQN = Multi-Agent Double Deep Q-Network; CMDP = Constrained Markov Decision Process; FL = Federated Learning; Prox.\ Reg.\ = Proximal Regularization; Confid.\ = Confidentiality; Eaves.\ = Eavesdropper; ES = Edge Server; AES = Advanced Encryption Standard; PLS = Physical Layer Security; PQC = Post-Quantum Cryptography; CKKS = Cheon-Kim-Kim-Song; ML-DSA = Module-Lattice Digital Signature Algorithm; Compl.\ Rate = Task Completion Rate; Multi Obj = Multi-Objective Optimization.
\end{minipage}
\end{table*}

\subsubsection{Research Gap Identification}
As shown in Table~\ref{tab:related_works}, five research gaps emerge.

\textbf{Gap 1: Oversimplified task modeling.} All existing works treat each vehicular task as a monolithic unit, whereas representative applications (e.g., autonomous-driving perception pipelines) comprise multiple dependent stages whose completion time is governed by the critical path rather than a sum of subtask delays. None jointly models subtask structure, dependency, and heterogeneity.

\textbf{Gap 2: No long-term constraint enforcement.} Task completion rate is safety-critical, yet no existing work guarantees a minimum long-term rate through a formal constraint-handling mechanism; where considered, it appears only as a soft reward term without any threshold guarantee.

\textbf{Gap 3: Federated learning without non-IID handling.} Only three works~\cite{10720029, 11085134, 11096732} adopt FL, and none addresses non-IID data: as traffic density, speed, and task composition vary across regions, standard aggregation lets local models drift from the global model, degrading aggregation quality and destabilizing training.

\textbf{Gap 4: Incomplete security without PQC-based parameter protection.} Although a few works secure the offloading process, each addresses only one of confidentiality and integrity~\cite{xiao2023joint, 10745593, ju2024eeso, ju2023sora, 11536917, 10969549}; even when confidentiality is provided, it guards only against external eavesdroppers, leaving the parameters exposed in plaintext at the edge servers that process them. The FL-based schemes~\cite{10720029, 11085134, 11096732}, by contrast, leave the exchanged model parameters unprotected altogether. Across all of them, none employs PQC to protect the parameters.

\textbf{Gap 5: Simulation-only evaluation.} All existing works evaluate purely in simulation, estimating delay and energy from analytical formulas rather than measuring real computation kernels on actual hardware.

CRAFT is designed to simultaneously address all five gaps by integrating CMDP-based Lagrangian constraint handling with DAG-aware scheduling, proximal-regularized federated aggregation with post-quantum security protection combining CKKS and ML-DSA signatures, and VM-based evaluation with real application workloads.

\section{System Model and Problem Formulation}
\label{sec:system_model}

This section presents the system model and formulates the optimization problem. The key notations are summarized in Appendix Table A1.

\subsection{System Overview}

As illustrated in Fig.~\ref{fig:system}, we consider a Vehicular Edge Computing (VEC) system consisting of multiple vehicles and multiple roadside Edge Servers (ESs). The system operates over a discretized time horizon divided into $T$ equal-length time slots, denoted by $\mathcal{T} = \{1, 2, \ldots, T\}$, each of duration $\tau$. We adopt a quasi-static model, where the channel conditions remain approximately constant within each time slot. Each DAG task is either completed or declared failed within the time slot in which it is generated, and the system begins a new decision epoch at the start of the next time slot.

The system comprises $N$ vehicles, denoted by $\mathcal{N} = \{1, 2, \ldots, N\}$, and $M$ ESs deployed on Roadside Units (RSUs), denoted by $\mathcal{M} = \{1, 2, \ldots, M\}$. Each ES $m$ has a coverage radius $R_m$ and a total computing capacity $F_m$ (cycles/s). Each vehicle generates computation tasks with internal dependencies (modeled as DAGs) while driving. A vehicle can either execute subtasks locally or offload them to an ES within its communication range via V2I links. The offloading destination, including the choice of which ES to use, is dynamically determined by the system rather than simply selecting the nearest ES.

\begin{figure}[t]
    \centering
    \includegraphics[width=\columnwidth]{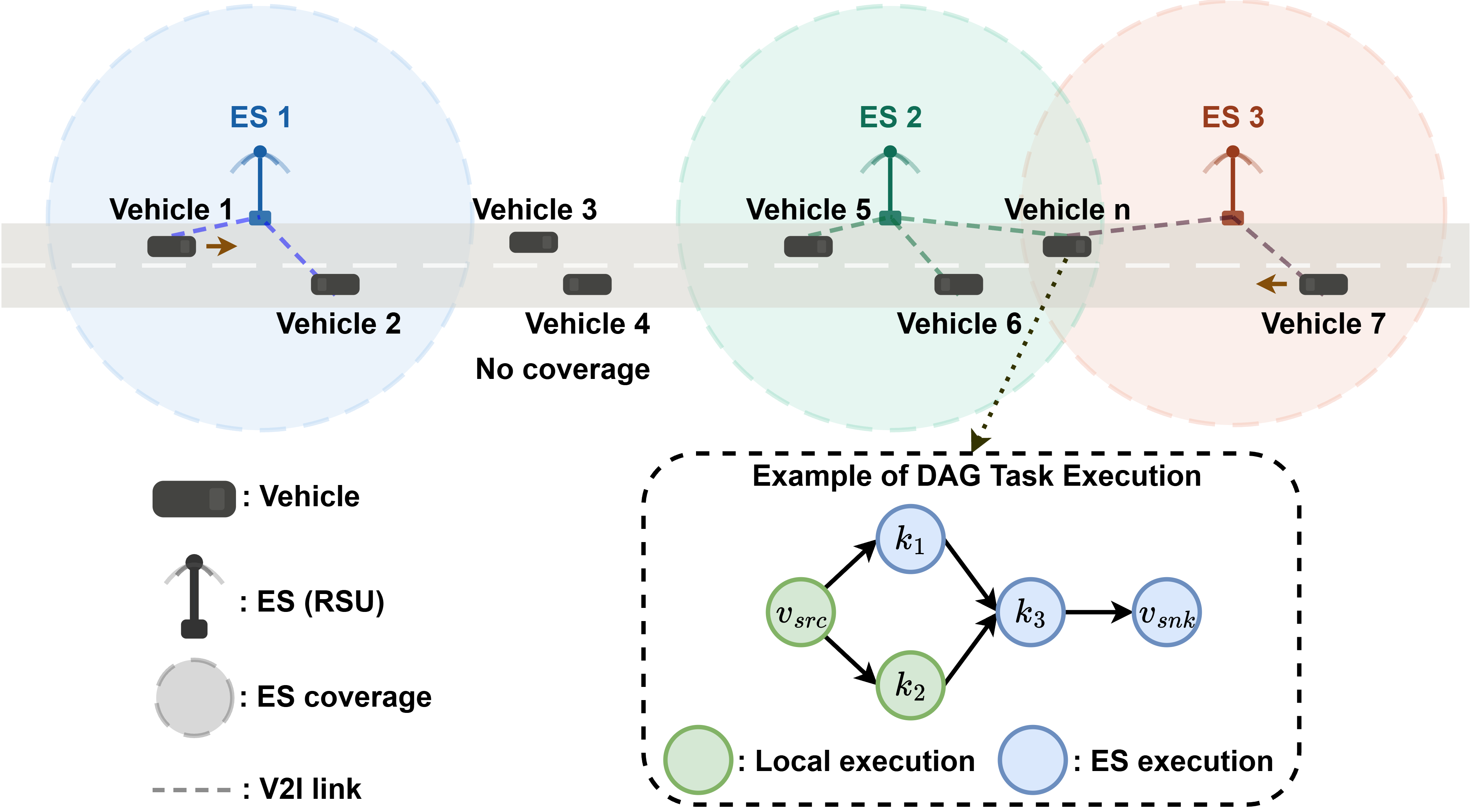}
    \caption{System overview of the VEC environment with multiple vehicles and ESs, where each vehicle's DAG task can be executed locally or offloaded to ESs.}
    \label{fig:system}
\end{figure}

\subsection{Mobility Model}

Vehicle trajectories are generated by Simulation of Urban Mobility (SUMO) on realistic road topologies. In time slot $t$, the position of vehicle $n$ is denoted as $\mathbf{p}_n(t) = (x_n(t), y_n(t))$, its velocity vector as $\mathbf{v}_n(t)$, and its speed as $v_n(t) = \|\mathbf{v}_n(t)\|$. ES $m$ is fixed at position $\mathbf{q}_m = (x_m, y_m)$. The distance between vehicle $n$ and ES $m$ is given by
\begin{equation}
    d_{n,m}(t) = \|\mathbf{p}_n(t) - \mathbf{q}_m\|_2.
\end{equation}

Vehicle $n$ is within the coverage of ES $m$ when $d_{n,m}(t) \leq R_m$. The set of ESs reachable by vehicle $n$ in time slot $t$ is $\mathcal{M}_n(t) = \{m \in \mathcal{M} \mid d_{n,m}(t) \leq R_m\}$.

Although the quasi-static model assumes constant channel parameters within a time slot, vehicles continue to move during the slot. Consequently, the residual sojourn time of a vehicle within an ES's coverage may be less than $\tau$. Assuming that the vehicle travels at constant velocity along its current heading during the estimation window, we define $\boldsymbol{\Delta}_{n,m}(t) = \mathbf{p}_n(t) - \mathbf{q}_m$. The residual sojourn time $D_{n,m}^{\mathrm{res}}(t)$ of vehicle $n$ with respect to ES $m$ is obtained by solving
\begin{equation}
    \|\boldsymbol{\Delta}_{n,m}(t) + \mathbf{v}_n(t) \cdot D_{n,m}^{\mathrm{res}}(t)\|^2 = R_m^2,
\end{equation}
yielding the positive root:
\begin{equation}
    \scalebox{0.85}{$D_{n,m}^{\mathrm{res}}(t) = \frac{-\boldsymbol{\Delta}_{n,m}(t) \cdot \mathbf{v}_n(t) + \sqrt{\left(\boldsymbol{\Delta}_{n,m}(t) \cdot \mathbf{v}_n(t)\right)^2 - \|\mathbf{v}_n(t)\|^2\left(\|\boldsymbol{\Delta}_{n,m}(t)\|^2 - R_m^2\right)}}{\|\mathbf{v}_n(t)\|^2}.$}
\end{equation}

When $D_{n,m}^{\mathrm{res}}(t) < \tau$, vehicle $n$ may exit the coverage of ES $m$ during the current time slot, and subtasks offloaded to ES $m$ face the risk of failure.

\subsection{Communication Model}

Vehicles communicate with ESs via V2I links. The system employs Orthogonal Frequency-Division Multiple Access (OFDMA)~\cite{3gpp_v2x} with a total bandwidth $B$, equally allocated among the vehicles connected to the same ES:
\begin{equation}
    b_{n,m}(t) = \frac{B}{|\mathcal{N}_m(t)|},
\end{equation}
where $\mathcal{N}_m(t)$ is the set of vehicles connected to ES $m$ in time slot $t$.

The channel gain between vehicle $n$ and ES $m$ is modeled as the product of path loss and small-scale Rayleigh fading:
\begin{equation}
    h_{n,m}(t) = h_0 \cdot d_{n,m}(t)^{-\alpha} \cdot |\xi_{n,m}(t)|^2,
\end{equation}
where $h_0$ is the channel gain at the reference distance, $\alpha$ is the path loss exponent, and $\xi_{n,m}(t) \sim \mathcal{CN}(0,1)$ is the Rayleigh fading coefficient.

The uplink transmission rate from vehicle $n$ to ES $m$ is
\begin{equation}
    r_{n,m}^{\mathrm{up}}(t) = b_{n,m}(t) \log_2 \left(1 + \frac{P_n^{\mathrm{tx}} \cdot h_{n,m}(t)}{\sigma^2}\right),
\end{equation}
where $P_n^{\mathrm{tx}}$ is the transmit power of vehicle $n$ and $\sigma^2$ is the additive white Gaussian noise power. The downlink transmission rate from ES $m$ to vehicle $n$ is
\begin{equation}
    r_{n,m}^{\mathrm{down}}(t) = b_{n,m}(t) \log_2 \left(1 + \frac{P_m^{\mathrm{ES}} \cdot h_{n,m}(t)}{\sigma^2}\right),
\end{equation}
where $P_m^{\mathrm{ES}}$ is the transmit power of ES $m$. Additionally, ESs are interconnected via a wired backhaul network with a transmission rate of $r^{\mathrm{wire}}$.

\subsection{Task Model}

Computation tasks of vehicle $n$ arrive according to a Poisson process with rate $\Lambda_n$. Each task is modeled as a DAG $G_n = (\mathcal{V}_n, \mathcal{E}_n)$. To provide a unified framework for modeling all data flows, we augment the original subtask set $\mathcal{K}_n = \{1, 2, \ldots, K_n\}$ with two virtual nodes: a source node $v_{\mathrm{src}}$ and a sink node $v_{\mathrm{snk}}$, both logically located at the vehicle with zero computation workload. The node set is $\mathcal{V}_n = \{v_{\mathrm{src}}\} \cup \mathcal{K}_n \cup \{v_{\mathrm{snk}}\}$.

Each actual subtask $k \in \mathcal{K}_n$ is characterized by a tuple $(d_k, c_k, o_k)$, where $d_k$ is the input data size (bits), $c_k$ is the required computation workload (CPU cycles), and $o_k$ is the output data size (bits).

The edge set $\mathcal{E}_n$ consists of three types of directed edges, each carrying a weight $w_{i,k}$ representing the amount of data to be transferred: (i) dependency edges between original subtasks $(i, k)$ with weight $w_{i,k} = o_i$; (ii) edges from $v_{\mathrm{src}}$ to each source subtask $k$ (i.e., subtasks with no predecessors) with weight $w_{v_{\mathrm{src}},k} = d_k$; and (iii) edges from each sink subtask $k$ (i.e., subtasks with no successors) to $v_{\mathrm{snk}}$ with weight $w_{k,v_{\mathrm{snk}}} = o_k$.

For each actual subtask $k \in \mathcal{K}_n$, the system makes an offloading decision $a_k \in \{0\} \cup \mathcal{M}_n(t)$, where $a_k = 0$ indicates local execution and $a_k = m$ indicates offloading to ES $m$. The virtual nodes are fixed at the vehicle: $a_{v_{\mathrm{src}}} = a_{v_{\mathrm{snk}}} = 0$.

\subsection{Computation Model}

\subsubsection{Local Computation}

If subtask $k$ is executed locally on vehicle $n$ ($a_k = 0$), the vehicle's CPU frequency is $f_n^{\mathrm{loc}}$ (cycles/s). The computation delay and energy consumption are
\begin{equation}
    D_k^{\mathrm{comp}} = \frac{c_k}{f_n^{\mathrm{loc}}}, \quad E_k^{\mathrm{comp}} = \kappa (f_n^{\mathrm{loc}})^2 \cdot c_k,
\end{equation}
where $\kappa$ is the effective capacitance coefficient related to the chip architecture.

\subsubsection{Edge Computation}

If subtask $k$ is offloaded to ES $m$ ($a_k = m$), the ES allocates computing resources $f_{k,m}^{\mathrm{edge}}$ (cycles/s) to it. The computation delay is:
\begin{equation}
    D_k^{\mathrm{comp}} = \frac{c_k}{f_{k,m}^{\mathrm{edge}}}.
\end{equation}
Edge execution does not consume vehicle-side energy, hence $E_k^{\mathrm{comp}} = 0$.

\subsubsection{Data Transfer Model}

Data transfers within the DAG are modeled through a unified framework. For each edge $(i, k) \in \mathcal{E}_n$, let $l_i$ and $l_k$ denote the execution locations of nodes $i$ and $k$, respectively ($l = 0$ for local, $l = m$ for ES $m$). The transfer delay $D_{i,k}^{\mathrm{trans}}$ and vehicle-side transfer energy $E_{i,k}^{\mathrm{trans}}$ are determined by the following four cases:

\noindent\textbf{(a)} $l_i = l_k$ (same node):
\begin{equation}
    D_{i,k}^{\mathrm{trans}} = 0, \quad E_{i,k}^{\mathrm{trans}} = 0.
\end{equation}

\noindent\textbf{(b)} $l_i = 0, l_k = m$ (local $\to$ ES $m$, uplink):
\begin{equation}
    D_{i,k}^{\mathrm{trans}} = \frac{w_{i,k}}{r_{n,m}^{\mathrm{up}}(t)}, \quad E_{i,k}^{\mathrm{trans}} = P_n^{\mathrm{tx}} \cdot D_{i,k}^{\mathrm{trans}}.
\end{equation}

\noindent\textbf{(c)} $l_i = m, l_k = 0$ (ES $m$ $\to$ local, downlink):
\begin{equation}
    D_{i,k}^{\mathrm{trans}} = \frac{w_{i,k}}{r_{n,m}^{\mathrm{down}}(t)}, \quad E_{i,k}^{\mathrm{trans}} = 0.
\end{equation}

\noindent\textbf{(d)} $l_i = m_1, l_k = m_2$, $m_1 \neq m_2$ (ES $\to$ ES, backhaul):
\begin{equation}
    D_{i,k}^{\mathrm{trans}} = \frac{w_{i,k}}{r^{\mathrm{wire}}}, \quad E_{i,k}^{\mathrm{trans}} = 0.
\end{equation}

This framework uniformly handles all data flow scenarios: the initial data upload for source subtasks corresponds to $v_{\mathrm{src}} \to k$ (case b), and the final result retrieval corresponds to $k \to v_{\mathrm{snk}}$ (case c).

\subsubsection{DAG Completion Time and Total Energy Consumption}

We define the predecessor set $\mathrm{pred}(k) = \{i \mid (i, k) \in \mathcal{E}_n\}$ and the successor set $\mathrm{succ}(k) = \{j \mid (k, j) \in \mathcal{E}_n\}$. The start time and finish time of subtask $k$ are recursively defined as
\begin{equation}
    D_k^{\mathrm{start}} = \max_{i \in \mathrm{pred}(k)} \left(D_i^{\mathrm{finish}} + D_{i,k}^{\mathrm{trans}}\right),
\end{equation}
\begin{equation}
    D_k^{\mathrm{finish}} = D_k^{\mathrm{start}} + D_k^{\mathrm{comp}},
\end{equation}
where $D_{v_{\mathrm{src}}}^{\mathrm{finish}} = 0$. The makespan of the entire DAG is
\begin{equation}
    D_n^{\mathrm{DAG}} = D_{v_{\mathrm{snk}}}^{\mathrm{start}}.
\end{equation}

The total vehicle-side energy consumption for a DAG task of vehicle $n$ is
\begin{equation}
    E_n^{\mathrm{total}} = \sum_{k \in \mathcal{K}_n} E_k^{\mathrm{comp}} + \sum_{(i,k) \in \mathcal{E}_n} E_{i,k}^{\mathrm{trans}}.
\end{equation}

\subsection{Problem Formulation}
\label{sec:problem}
Our objective is to jointly optimize the offloading decisions $\mathbf{a}(t) = \{a_k(t)\}$ and the ES resource allocation $\mathbf{f}(t) = \{f_{k,m}^{\mathrm{edge}}(t)\}$ across all vehicles and time slots, so as to minimize a weighted sum of system-wide latency and vehicle-side energy consumption while ensuring that the task completion rate is no less than a given threshold.

\textbf{Task Success Condition.} The DAG task of vehicle $n$ in time slot $t$ is deemed successful ($s_n(t) = 1$) if and only if both of the following conditions are satisfied; otherwise, $s_n(t) = 0$.

\emph{Condition 1 (Deadline):} DAG is completed within deadline:
\begin{equation}
    D_n^{\mathrm{DAG}}(t) \leq D_n^{\max}.
\end{equation}

\emph{Condition 2 (Sojourn Coverage):} For every subtask $k$ offloaded to ES $m$ ($a_k = m$), vehicle $n$ remains within the coverage of ES $m$ throughout the entire interaction, from data upload through computation to result retrieval. Formally, the last interaction time of subtask $k$ at ES $m$ is defined as
\begin{equation}
    \hat{D}_k(t) = D_k^{\mathrm{finish}}(t) + \max_{j \in \mathrm{succ}(k): l_j = 0} D_{k,j}^{\mathrm{trans}}(t),
\end{equation}
where the second term is zero when all successors are also executed on ES $m$. This condition requires
\begin{equation}
    \hat{D}_k(t) \leq D_{n,m}^{\mathrm{res}}(t), \quad \forall k \in \mathcal{K}_n \text{ with } a_k(t) = m.
\end{equation}

The optimization problem is formulated as follows:
\begin{subequations}\label{eq:P1}
{\small
\begin{alignat}{2}
    \mathbf{P1}:\min_{\mathbf{a}(t), \mathbf{f}(t)} \ &\frac{1}{NT} \sum_{t=1}^{T} \sum_{n=1}^{N} \left[\omega_1 \cdot D_n^{\mathrm{DAG}}(t) + \omega_2 \cdot E_n^{\mathrm{total}}(t)\right] \label{eq:obj}\\
    \mathrm{s.t.} \quad & \text{C1}: \frac{1}{NT} \sum_{t=1}^{T} \sum_{n=1}^{N} s_n(t) \geq \theta, \label{eq:C1}\\
    & \text{C2}: \!\sum_{k: a_k(t) = m}\! f_{k,m}^{\mathrm{edge}}(t) \leq F_m, \  \forall m \!\in\! \mathcal{M}, \forall t, \label{eq:C2}\\
    & \text{C3}: D_k^{\mathrm{start}}(t) \geq D_i^{\mathrm{finish}}(t) + D_{i,k}^{\mathrm{trans}}(t), \nonumber\\
    & \qquad\ \forall (i,k) \!\in\! \mathcal{E}_n, \forall n, \forall t, \label{eq:C3}\\
    & \text{C4}: a_k(t) \!\in\! \{0\} \cup \mathcal{M}_n(t), \  \forall k \!\in\! \mathcal{K}_n, \forall n, \forall t, \label{eq:C4}\\
    & \text{C5}: f_{k,m}^{\mathrm{edge}}(t) > 0, \  \forall k \text{ with } a_k(t) \!=\! m. \label{eq:C5}
\end{alignat}
}
\end{subequations}
where $\omega_1, \omega_2 > 0$ are the weighting coefficients, and $\theta \in (0,1]$ is the minimum task completion rate threshold.

C1 is a long-term statistical constraint requiring the overall task completion rate to be no less than $\theta$, where the success condition encompasses both the deadline and the sojourn coverage requirements. Unlike per-slot hard constraints, C1 permits occasional task failures in individual time slots as long as the long-term average success rate meets the threshold. C2 is the ES computing resource capacity constraint. C3 is the DAG dependency constraint, ensuring that a subtask begins execution only after all its predecessors have completed and their output data have been received. C4 restricts the offloading target to currently reachable ESs. C5 ensures each offloaded subtask receives a positive amount of computing resources.

\textbf{P1} is a Mixed-Integer Nonlinear Programming (MINLP) problem, where the offloading decisions $\mathbf{a}$ are discrete variables and the resource allocations $\mathbf{f}$ are continuous variables. Solving this problem poses three key challenges: (i) coupling of discrete and continuous variables causes the search space to grow exponentially with the number of vehicles and subtasks, rendering the problem NP-hard in general; (ii) the time-varying nature of vehicle mobility and channel conditions necessitates online real-time decision-making; and (iii) C1 is a long-term statistical constraint spanning multiple time slots, which cannot be verified within a single slot. Given these challenges, \textbf{P1} is naturally modeled as a Constrained Markov Decision Process (CMDP), where C1 is handled via Lagrangian relaxation and transformed into a dual optimization problem. In Section~\ref{sec:algorithm}, we propose CRAFT, a secure federated multi-agent Reinforcement Learning framework with Lagrangian-based constraint handling, to solve the problem online.

\section{Proposed CRAFT Framework}
\label{sec:algorithm}

This section presents the CRAFT framework for solving \textbf{P1} online. All notations follow from Section~\ref{sec:system_model}.

\subsection{Framework Overview}

To solve the CMDP problem formulated in Section~\ref{sec:problem}, we propose CRAFT. As shown in Fig.~\ref{fig:framework}, CRAFT consists of four cooperating components: (i) a multi-agent decision module following the CTDE paradigm, where Actors are deployed on vehicles for offloading decisions and Critics are deployed on ESs for value estimation; (ii) a Lagrangian relaxation-based constraint handling module incorporating the long-term completion rate constraint into the reward function via online dual variable updates; (iii) a constraint-aware federated Actor aggregation module that enables cross-region knowledge transfer through a two-tier aggregation structure with proximal regularization for non-IID resilience and Lagrangian stabilization; and (iv) a post-quantum security module that combines CKKS for parameter confidentiality during aggregation against both eavesdroppers and honest-but-curious edge servers, with ML-DSA post-quantum digital signatures for integrity verification.

\begin{figure}[t]
    \centering
    \includegraphics[width=\columnwidth, height=7.5cm]{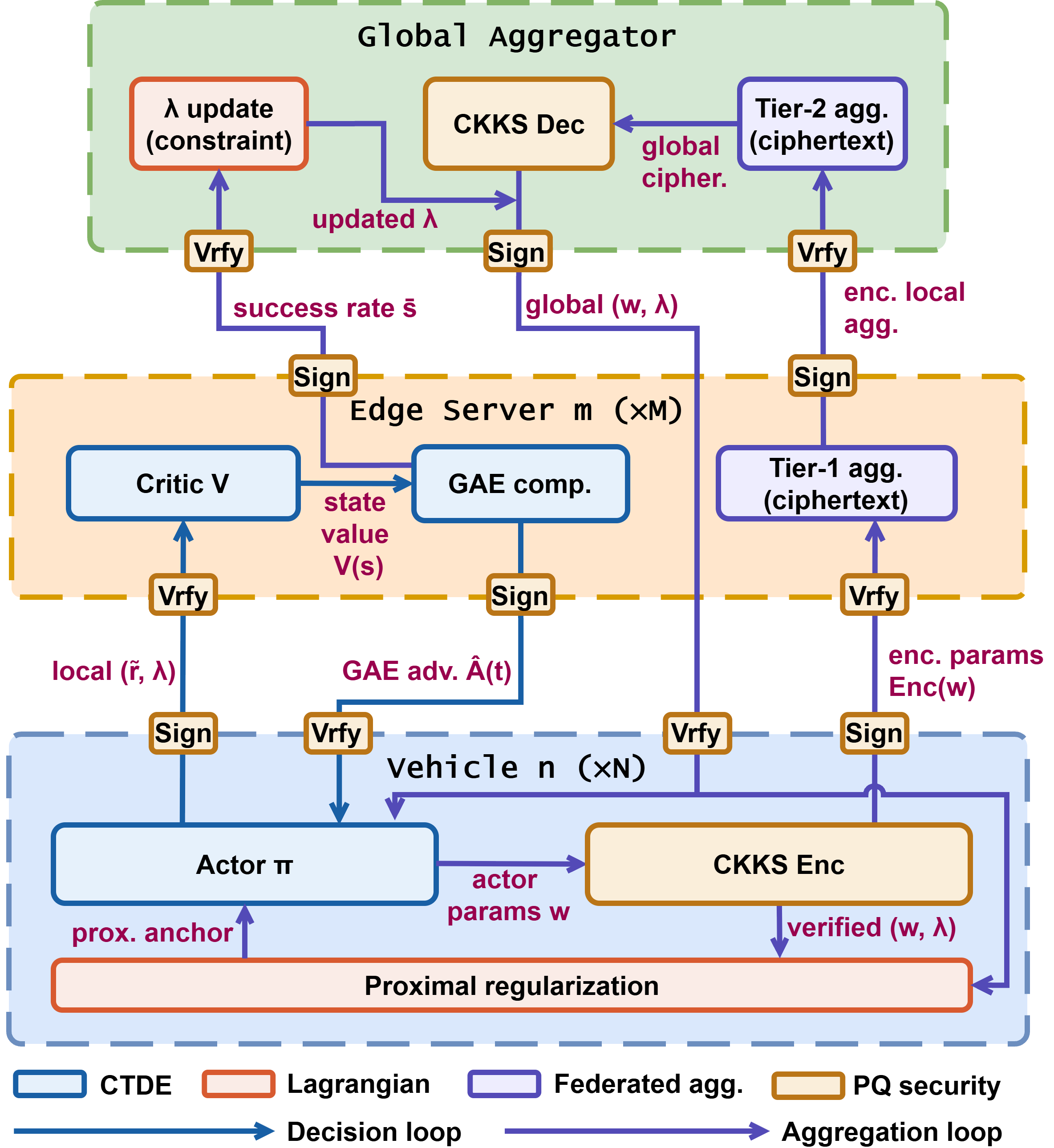}
    \caption{Overview of the CRAFT framework.}
    \label{fig:framework}
\end{figure}

CRAFT operates two parallel loops. The \emph{decision loop} executes once per time slot: vehicles make offloading decisions based on local observations, ESs compute advantage estimates using global states and return them to vehicles, and vehicles update their local Actors accordingly. The \emph{aggregation loop} executes every $K$ time slots: vehicles upload encrypted Actor parameters to ESs, which are aggregated through two tiers and redistributed to all vehicles, with the Lagrangian multiplier synchronized simultaneously.

\subsection{CMDP Formulation}
\label{sec:cmdp}
We model \textbf{P1} as a Constrained Markov Decision Process (CMDP), defined as follows.

\textbf{State Space.} Since the Actor and Critic in the CTDE paradigm require different inputs, we define the local observation and global state separately. The local observation of vehicle $n$ in time slot $t$ is
\begin{align}
    o_n(t) = \big\{ & \mathbf{p}_n(t),\ \mathbf{v}_n(t),\ \{d_k, c_k, o_k, \mathcal{E}_k\}_{k \in \mathcal{K}_n},\ f_n^{\mathrm{loc}}, \nonumber\\
    & \{d_{n,m}(t),\ D_{n,m}^{\mathrm{res}}(t)\}_{m \in \mathcal{M}_n(t)} \big\},
\end{align}
which includes the vehicle's position and velocity, the subtask characteristics and dependency structure of the current DAG task, the local computing capacity, and the distance and residual sojourn time to each reachable ES.

The global state maintained by ES $m$ is
\begin{equation}
    \scalebox{0.85}{$s_m(t) = \left\{ \{o_n(t)\}_{n \in \mathcal{N}_m(t)}, \{F_{m^{\prime}}^{\mathrm{avail}}(t)\}_{m^{\prime} \in \mathcal{M}}, \{|\mathcal{N}_{m^{\prime}}(t)|\}_{m^{\prime} \in \mathcal{M}}, \lambda(t) \right\},$}
\end{equation}
which includes the local observations of all vehicles within the coverage, the available computing resources and the number of connected vehicles at each ES, and the current Lagrangian multiplier.

\textbf{Action Space.} The action of vehicle $n$ comprises two parts, forming a hybrid action space:
\begin{equation}
    a_n(t) = \left( \{a_k\}_{k \in \mathcal{K}_n},\ \{\beta_k\}_{k \in \mathcal{K}_n: a_k \neq 0} \right),
\end{equation}
where $a_k \in \{0\} \cup \mathcal{M}_n(t)$ is the discrete offloading decision for subtask $k$, and $\beta_k \in (0, 1]$ is the resource request weight for subtask $k$. Since multiple vehicles may simultaneously request resources from the same ES, the actual allocated computing frequency is determined by proportional normalization at the ES side:
\begin{equation}
\label{eq:resource_alloc}
    f_{k,m}^{\mathrm{edge}}(t) = \frac{\beta_k}{\sum_{n \in \mathcal{N}_m(t)}
    \sum_{k' \in \mathcal{K}_n:\, a_{k'}=m} \beta_{k'}} \cdot F_m,
\end{equation}
which automatically satisfies the capacity constraint C2 in~\eqref{eq:C2}. The Actor network outputs offloading decisions through a discrete head (softmax) and resource request weights through a continuous head (Gaussian distribution).

\textbf{Reward Function.} Since latency and energy consumption have different physical units and magnitudes, we apply min-max normalization to map both metrics to $[0, 1]$:
\begin{equation}
    \bar{D}_n(t) = \frac{D_n^{\mathrm{DAG}}(t) - D_{\min}}{D_{\max} - D_{\min}}, \quad \bar{E}_n(t) = \frac{E_n^{\mathrm{total}}(t) - E_{\min}}{E_{\max} - E_{\min}},
\end{equation}
where $D_{\max}, D_{\min}, E_{\max}, E_{\min}$ are the maximum and minimum values observed across all vehicles and time slots during training. The reward of vehicle $n$ in time slot $t$ is
\begin{equation}
    r_n(t) = -\left[\omega_1 \cdot \bar{D}_n(t) + \omega_2 \cdot \bar{E}_n(t)\right].
\end{equation}

\textbf{Constraint Cost.} Corresponding to the long-term constraint C1, the per-step constraint cost is defined as
\begin{equation}
    c_n(t) = \theta - s_n(t),
\end{equation}
where $s_n(t)$ is the task success indicator defined in Section~\ref{sec:problem}. C1 is equivalent to requiring the long-term average constraint cost to be non-positive: $\frac{1}{NT}\sum_{t,n} c_n(t) \leq 0$.

\textbf{State Transition.} The state transition is jointly driven by vehicle mobility (generated by SUMO) and new DAG task arrivals (Poisson process), satisfying the Markov property.

\subsection{Lagrangian-based Constrained Optimization}
\label{sec:lagrangian}

The choice of CMDP with Lagrangian relaxation is motivated by three structural properties of C1. First, C1 is a trajectory-level statistical constraint that cannot be evaluated at individual states, ruling out per-step safe action filtering methods~\cite{ames2017cbf}. Second, the binary indicator $s_n(t)$ is non-differentiable, precluding direct application of trust-region constrained policy methods such as Constrained Policy Optimization (CPO)~\cite{achiam2017cpo} that rely on smooth constraint surrogates. Third, fixed-penalty reward shaping lacks an adaptive mechanism for correcting long-term constraint violations~\cite{tessler2018reward}. The CMDP framework~\cite{altman1999cmdp} naturally accommodates such long-term constraints, and Lagrangian relaxation provides a principled way to adaptively penalize task failures while remaining compatible with actor-critic architectures and federated aggregation. In Section~\ref{sec:analysis}, we further provide conditional stability bounds showing how the projected dual update controls the average positive constraint violation up to a residual term induced by inexact non-convex policy optimization.

Specifically, introducing the Lagrangian multiplier $\lambda \geq 0$, we construct the Lagrangian function:
\begin{equation}
    \mathcal{L}(\pi, \lambda) = \mathbb{E}_\pi\left[\frac{1}{NT}\sum_{t=1}^{T}\sum_{n=1}^{N}\left(-r_n(t) + \lambda \cdot c_n(t)\right)\right].
\end{equation}

The original problem \textbf{P1} is transformed into:
\begin{equation}
    \min_\pi \max_{\lambda \geq 0}\ \mathcal{L}(\pi, \lambda).
\end{equation}

The dual problem is solved via alternating optimization. In the inner loop, given $\lambda$, the policy $\pi$ is optimized by maximizing the augmented reward:
\begin{align}\label{eq:aug_reward}
    \tilde{r}_n(t) &= r_n(t) - \lambda \cdot c_n(t) \nonumber\\
    &= -\left[\omega_1 \cdot \bar{D}_n(t) + \omega_2 \cdot \bar{E}_n(t)\right] - \lambda \cdot (\theta - s_n(t)).
\end{align}

When a task fails ($s_n(t) = 0$), the augmented reward incurs a penalty of $-\lambda\theta$; when a task succeeds ($s_n(t) = 1$), it receives a compensation of $\lambda(1-\theta)$. A larger $\lambda$ imposes heavier penalties on task failures, driving the policy toward more conservative decisions to improve the completion rate.

In the outer loop, $\lambda$ is updated online via gradient ascent:
\begin{equation}\label{eq:lambda_update}
    \lambda \leftarrow \Pi_{[0,\,\lambda^{\max}]}\!\left(\lambda + \eta_\lambda \cdot (\theta - \bar{s}(t))\right),
\end{equation}
where $\eta_\lambda > 0$ is the multiplier learning rate, $\bar{s}(t)$ is the recent average task success rate, $\lambda^{\max}$ is a prescribed upper bound on the multiplier, and $\Pi_{[0,\,\lambda^{\max}]}$ denotes projection onto $[0, \lambda^{\max}]$. When the success rate falls below $\theta$, $\lambda$ increases to intensify the failure penalty; conversely, $\lambda$ decreases to allow the policy to pursue lower latency and energy consumption. This mechanism enables dynamic balancing between performance optimization and constraint satisfaction.

\subsection{Multi-Agent Policy Optimization}
\label{sec:mappo}

This subsection details the multi-agent decision-making and training mechanisms.

\textbf{Architecture Design.} CRAFT adopts the CTDE paradigm, whose suitability arises from the inherent architecture of VEC systems: ESs naturally acquire state information of all vehicles within their coverage through the task offloading management process, and can therefore serve as centralized Critics with zero additional communication overhead. Specifically, the Actor network $\pi_{\boldsymbol{w}_n}$ is deployed on vehicle $n$ and outputs actions based on local observation $o_n(t)$; the Critic network $V_{\boldsymbol{w}_m^c}$ is deployed on ES $m$ and outputs state value estimates based on global state $s_m(t)$.

\textbf{Network Structure.} Local observation $o_n(t)$ contains variable-length components: the subtask count $K_n$ varies across task types and the reachable ES set $\mathcal{M}_n(t)$ changes as vehicles move. Each subtask's features $(d_k, c_k, o_k)$ and dependency context are encoded by a shared embedding layer. A learnable aggregation token is prepended to the subtask embeddings and processed through a self-attention layer; the output at the aggregation token position serves as the fixed-dimension DAG representation. ES-related features span all $M$ ESs with unreachable entries masked to zero. The Actor combines the DAG representation with masked ES features through a fully connected layer, followed by two parallel output heads. The discrete head outputs a probability distribution over $M+1$ offloading options for each subtask via softmax, with unreachable ESs masked to zero and renormalized before sampling: $\pi_{\boldsymbol{w}_n}^{\mathrm{disc}}(a_k \mid o_n(t))$. Continuous head outputs resource request weights via a Gaussian distribution: $\beta_k \sim \mathcal{N}(\mu_{\boldsymbol{w}_n} (o_n(t)),\ \sigma_{\boldsymbol{w}_n}^2(o_n(t)))$, mapped to $(0, 1]$ through a sigmoid function. The Critic adopts same aggregation token mechanism over vehicle embeddings: each vehicle's observation within ES $m$'s coverage is encoded by a shared embedding layer, and self-attention with an aggregation token produces a fixed-dimension state representation independent of $|\mathcal{N}_m(t)|$.

\textbf{Training Procedure.} Training proceeds through the following steps in each time slot $t$:

\emph{Step 1:} Vehicle $n$ samples action $a_n(t) \sim \pi_{\boldsymbol{w}_n}(\cdot \mid o_n(t))$ based on local observation $o_n(t)$, obtaining offloading decisions $\{a_k\}$ and resource request weights $\{\beta_k\}$.

\emph{Step 2:} Each ES $m$ collects the offloading decisions and resource request weights from all vehicles within its coverage and computes the actual resource allocation via proportional normalization~\eqref{eq:resource_alloc}. The allocated frequencies are returned to the respective vehicles.

\emph{Step 3:} Vehicles execute their offloading decisions with the allocated resources and compute the augmented reward $\tilde{r}_n(t)$ via Eq.~\eqref{eq:aug_reward}. The vehicle uploads this scalar to its serving ES with an ML-DSA signature for integrity verification (Section~\ref{sec:fl}). The environment returns the observation for the next time slot.

\emph{Step 4:} ES $m$ stores the experience tuple $(o_n, a_n, \tilde{r}_n, o_n')$ in its buffer. After vehicle $n$ accumulates $T_{\mathrm{batch}}$ consecutive time slots of experience within ES $m$'s coverage, the ES computes the Generalized Advantage Estimation (GAE)~\cite{schulman2016gae}:
\begin{equation}
\label{eq:gae}
    \hat{A}_n^{(t)} = \sum_{l=0}^{T_{\mathrm{batch}}-t-1}(\gamma\varsigma)^l \delta_n^{(t+l)},
\end{equation}
where $\delta_n^{(t)} = \tilde{r}_n(t) + \gamma V_{\boldsymbol{w}_m^c}(s_m(t+1)) - V_{\boldsymbol{w}_m^c}(s_m(t))$ is the temporal difference error, $\gamma$ is the discount factor, and $\varsigma$ is the GAE parameter. If a vehicle exits the coverage of ES $m$ before accumulating $T_{\mathrm{batch}}$ steps, the trajectory is truncated at the boundary and the Critic's value estimate $V_{\boldsymbol{w}_m^c}(s_m(t+1))$ is used as a bootstrap estimate for the remaining returns. The ES returns the computed advantage value $\hat{A}_n^{(t)}$ (a scalar) to the vehicle, also protected by an ML-DSA signature (Section~\ref{sec:fl}).

\emph{Step 5:} The vehicle updates its Actor parameters using a clipped surrogate objective~\cite{schulman2017ppo}:
\begin{equation}
\label{eq:clip}
    \scalebox{0.8}{$\mathcal{L}^{\mathrm{clip}}(\boldsymbol{w}_n) = -\mathbb{E}\left[\min\left(\rho_n(t)\hat{A}_n^{(t)},\ \mathrm{clip}(\rho_n(t), 1\!-\!\epsilon, 1\!+\!\epsilon)\hat{A}_n^{(t)}\right)\right].$}
\end{equation}
where $\rho_n(t) = \frac{\pi_{\boldsymbol{w}_n}(a_n(t) \mid o_n(t))}{\pi_{\boldsymbol{w}_n^{\mathrm{old}}}(a_n(t) \mid o_n(t))}$ is the importance sampling ratio and $\epsilon$ is the clipping parameter. The clipping mechanism constrains the magnitude of policy updates, which ensures training stability and makes the policy robust to variations in Critic estimates when vehicles cross ES coverage boundaries.

\emph{Step 6:} ES $m$ updates its Critic parameters using the collected experience:
\begin{equation}\label{eq:critic}
    \mathcal{L}^{\mathrm{Critic}}(\boldsymbol{w}_m^c) = \mathbb{E}\left[\left(V_{\boldsymbol{w}_m^c}(s_m(t)) - \hat{R}_n^{(t)}\right)^2\right],
\end{equation}
where $\hat{R}_n^{(t)} = \hat{A}_n^{(t)} + V_{\boldsymbol{w}_m^{c,\mathrm{old}}}(s_m(t))$ is the target return.

When a vehicle is not within the coverage of any ES, its action space degenerates to executing all subtasks locally, and the training step is skipped.

\subsection{Secure Federated Actor Aggregation}
\label{sec:fl}

Although the CTDE-based multi-agent training enables effective local policy learning within each ES region, independent training across different regions may lead to policy fragmentation, where vehicles in high-density areas, low-density areas, urban roads, and highways learn significantly different strategies. To address this, CRAFT introduces FL for cross-region knowledge transfer. FL aggregates only the Actor parameters and not the Critic, since the Critic operates on global states whose semantics and dimensions vary across ES regions, making aggregation meaningless.

\textbf{Two-Tier Aggregation Structure.} Every $K$ time slots, the system performs one round of FL aggregation consisting of two tiers.

The first tier is intra-ES aggregation. ES $m$ collects Actor parameters from the set of participating vehicles $\mathcal{N}_m^{\mathrm{FL}}(t)$ within its coverage and computes locally aggregated model:
\begin{equation}
\label{eq:tier1_agg}
    \boldsymbol{w}_m^{\mathrm{local}} = \sum_{n \in \mathcal{N}_m^{\mathrm{FL}}(t)} \frac{|\mathcal{D}_n|}{\sum_{n' \in \mathcal{N}_m^{\mathrm{FL}}(t)}|\mathcal{D}_{n'}|} \boldsymbol{w}_n,
\end{equation}
where $|\mathcal{D}_n|$ is the number of experience samples collected by vehicle $n$ since the last aggregation round.

The second tier is inter-ES aggregation. A trusted global aggregator collects the locally aggregated results from all ESs via the backhaul network to compute the global model:
\begin{equation}
\label{eq:tier2_agg}
    \boldsymbol{w}^{\mathrm{global}} = \sum_{m \in \mathcal{M}} \frac{|\mathcal{N}_m^{\mathrm{FL}}(t)|}{\sum_{m' \in \mathcal{M}}|\mathcal{N}_{m'}^{\mathrm{FL}}(t)|} \boldsymbol{w}_m^{\mathrm{local}}.
\end{equation}

The global model is then distributed to vehicles as the initial parameters for the next round of local training.

\textbf{Constraint-Aware Non-IID Federated Aggregation.}
In VEC scenarios, experience data collected across different regions exhibit significant non-IID characteristics, with heterogeneity arising from three dimensions. First, \emph{spatial heterogeneity}: traffic density variations across regions lead to substantially different state space distributions, e.g., urban centers feature dense vehicle populations with intense resource contention, whereas peripheral road segments have sparse vehicles with abundant ES resources. Second, \emph{task distribution heterogeneity}: the dominant task types differ across regions, e.g., commercial districts are dominated by augmented reality navigation tasks while highway segments primarily generate autonomous driving perception tasks, causing systematic bias in the task characteristic dimensions of each vehicle's collected experience. Third, \emph{temporal non-stationarity}: vehicles continuously enter and exit the coverage of different ESs, causing participants in each aggregation round to change dynamically, further exacerbating the instability of data distributions.

Under these strong non-IID conditions, standard Federated Averaging (FedAvg) aggregation allows local models to deviate substantially from the global model between aggregation rounds, destabilizing the quality of the aggregated global model. We adopt the proximal regularization mechanism from Federated Proximal (FedProx)\cite{li2020fedprox} to address this challenge, rather than alternative approaches such as Stochastic Controlled Averaging for Federated Learning (SCAFFOLD) or Federated Normalized Averaging (FedNova), based on the following considerations. SCAFFOLD requires transmitting additional control variates in each aggregation round, doubling the V2I uplink communication volume, which incurs prohibitive overhead in bandwidth-constrained vehicular networks. FedNova requires explicit normalization over the number of local update steps across participants, but in VEC scenarios, vehicles' local update counts naturally vary due to mobility and are difficult to predict, resulting in high implementation complexity. FedProx requires only adding a proximal term to local loss function with no additional communication overhead, and inherently tolerates heterogeneous local update steps.

Incorporating the proximal term into CRAFT's local training objective yields the complete loss function:
\begin{equation}\label{eq:local_loss}
    \mathcal{L}^{\mathrm{local}}(\boldsymbol{w}_n) = \mathcal{L}^{\mathrm{clip}}(\boldsymbol{w}_n) + \frac{\mu}{2}\|\boldsymbol{w}_n - \boldsymbol{w}^{\mathrm{global}}\|^2,
\end{equation}
where $\mathcal{L}^{\mathrm{clip}}$ is the clipped surrogate objective based on the augmented reward $\tilde{r}_n(t)$ (which already incorporates the Lagrangian penalty described in Section~\ref{sec:lagrangian}), and $\mu > 0$ is the proximal coefficient.

Notably, within CRAFT's constrained optimization framework, the proximal term serves a dual role. Its first role is the classical non-IID drift control: by anchoring local models near the global model, it prevents policies across different regions from diverging into mutually incompatible local optima due to heterogeneous experience data. Its second role is acting as a stabilizer for the Lagrangian dynamic adjustment. In CRAFT, the Lagrangian multiplier $\lambda$ is updated dynamically based on the task completion rate: when the completion rate falls below the threshold, $\lambda$ increases to impose heavier penalties on task failures. However, a sudden increase in $\lambda$ alters the optimization objective of all vehicles through the augmented reward $\tilde{r}_n(t)$, driving all local models to simultaneously shift toward conservative strategies with large parameter updates. Without the proximal term, such a global policy shift would cause local models to deviate substantially from the current global model before aggregation, resulting in severe oscillations in the aggregated global model. The proximal term constrains the magnitude of parameter changes during each round of local updates, imposing a damping effect on the policy adjustments triggered by $\lambda$ variations and ensuring a smooth constraint adaptation process. The effectiveness of this dual role is validated through the ablation study in Section~\ref{sec:experiments}.

\textbf{Post-Quantum Security Protection.}
Throughout CRAFT's training pipeline, data is continuously exchanged among vehicles, ESs, and the global aggregator. These communication links face two categories of security threats: (i) \emph{privacy leakage}, where honest-but-curious ESs or external eavesdroppers may infer vehicle driving behavioral patterns and preferences from plaintext model parameters; and (ii) \emph{integrity violation}, where man-in-the-middle attackers can intercept and tamper with data in transit, e.g., injecting noise into model parameters or systematically corrupting advantage values to induce cumulative policy drift that later propagates through federated aggregation. CRAFT addresses these threats through a layered post-quantum security design: CKKS protects the confidentiality of model parameters during aggregation and enables computation over encrypted data, while ML-DSA digital signatures protect the integrity of all exchanged data. Both primitives are built on lattice-based hardness assumptions, yielding a fully post-quantum secure training pipeline.

\emph{Key setup and enrollment.} We assume a standard vehicular Public Key Infrastructure (PKI), as adopted in Vehicle-to-Everything (V2X) security credential management systems~\cite{yoshizawa2023survey}: a trusted, offline Certificate Authority (CA) enrolls each legitimate vehicle and issues a certificate binding the vehicle's ML-DSA public key $\mathit{pk}_i$ to its authorized identity. Each entity signs with its own private key $\mathit{sk}_i$, and a recipient accepts a message only if (i) the accompanying certificate is valid and unrevoked under the CA, and (ii) the ML-DSA signature verifies under the certified $\mathit{pk}_i$. Consequently, a self-generated key pair without a valid CA certificate is rejected, and impersonating an enrolled vehicle requires its private key, which never leaves the vehicle. The trusted global aggregator additionally holds a single CKKS key pair $(\mathit{pk}^{\mathrm{HE}}, \mathit{sk}^{\mathrm{HE}})$, distributing $\mathit{pk}^{\mathrm{HE}}$ to all enrolled vehicles and retaining $\mathit{sk}^{\mathrm{HE}}$ as the sole decryption key.

\emph{Confidentiality: CKKS Homomorphic Encryption.}
Conventional encryption schemes (e.g., AES) protect parameters only against network eavesdroppers: since one cannot compute over their ciphertexts, the edge server must decrypt before aggregating, exposing every vehicle's plaintext parameters to an honest-but-curious server. CKKS instead supports computation over ciphertexts, so the two-tier weighted averaging in~\eqref{eq:tier1_agg}--\eqref{eq:tier2_agg} runs entirely in the encrypted domain and no edge server ever observes individual plaintext parameters, providing confidentiality against both eavesdroppers and edge servers. Beyond this homomorphic property, the selection of CKKS~\cite{cheon2017ckks} is motivated by three further considerations. First, neural network parameters are floating-point values, and CKKS natively supports approximate arithmetic over real numbers, whereas exact HE schemes such as Brakerski/Fan-Vercauteren (BFV) and Brakerski-Gentry-Vaikuntanathan (BGV)~\cite{fan2012somewhat, brakerski2014leveled} require additional fixed-point quantization that introduces unnecessary precision loss and implementation complexity. Second, compared with Differential Privacy, CKKS does not inject artificial noise into model parameters, thereby avoiding the privacy-accuracy trade-off. Third, compared with Secure Aggregation protocols, CKKS does not require all participants to be simultaneously online, making it better suited for VEC scenarios where vehicles dynamically enter and exit ES coverage.

Since CKKS supports addition and scalar multiplication on ciphertexts, the two-tier weighted averaging defined in~\eqref{eq:tier1_agg} and~\eqref{eq:tier2_agg} can be performed entirely in the ciphertext domain. In each FL aggregation round, every participating vehicle encrypts its Actor parameters under the aggregator's public key $\mathit{pk}^{\mathrm{HE}}$: $\boldsymbol{e}_n = \mathrm{CKKS\_Enc}(\mathit{pk}^{\mathrm{HE}}, \boldsymbol{w}_n)$. The ES and the global aggregator then complete the two-tier aggregation in the ciphertext domain:
\begin{align}
\label{eq:cipher_agg_tier1}
    \boldsymbol{e}_m^{\mathrm{local}} &= \sum_{n \in \mathcal{N}_m^{\mathrm{FL}}} \frac{|\mathcal{D}_n|}{\sum_{n' \in \mathcal{N}_m^{\mathrm{FL}}}|\mathcal{D}_{n'}|} \cdot \boldsymbol{e}_n, \\
\label{eq:cipher_agg_tier2}
    \boldsymbol{e}^{\mathrm{global}} &= \sum_{m \in \mathcal{M}} \frac{|\mathcal{N}_m^{\mathrm{FL}}|}{\sum_{m' \in \mathcal{M}}|\mathcal{N}_{m'}^{\mathrm{FL}}|} \cdot \boldsymbol{e}_m^{\mathrm{local}}.
\end{align}
The global aggregator decrypts $\boldsymbol{e}^{\mathrm{global}}$ with its private key $\mathit{sk}^{\mathrm{HE}}$ to recover the plaintext aggregation result of~\eqref{eq:tier1_agg}--\eqref{eq:tier2_agg}. Throughout this process, ESs operate solely on ciphertexts and can never access plaintext parameters, thereby preserving individual vehicles' parameter privacy.

\emph{Integrity: ML-DSA Post-Quantum Digital Signatures.} CRAFT deploys ML-DSA digital signatures on all communication links to protect data integrity. ML-DSA is a lattice-based post-quantum signature scheme built on the Module Learning with Errors (MLWE) problem, sharing the same cryptographic foundation as CKKS. ML-DSA is chosen over classical ones (e.g., Elliptic Curve Digital Signature Algorithm (ECDSA)), preventing overall security from being downgraded by the quantum vulnerability of a classical signature scheme.

The signing protocol applies uniformly to all transmissions. For any payload $\mathbf{z}$ transmitted from entity $i$ to another entity, the sender holding its private key $\mathit{sk}_i$ computes the signature
\begin{equation}
    \sigma^{\mathbf{z}} = \mathrm{ML\text{-}DSA\_Sign}\!\left(\mathit{sk}_i,\, \mathcal{H}(\mathbf{z})\right),
\end{equation}
and transmits $(\mathbf{z},\, \sigma^{\mathbf{z}})$. Here, $\mathcal{H}(\cdot)$ denotes a cryptographic hash function (e.g., Secure Hash Algorithm 3 (SHA-3)). The receiver, holding the sender's certified public key $\mathit{pk}_i$, verifies before use:
\begin{equation}
    \mathrm{ML\text{-}DSA\_Verify}\!\left(\mathit{pk}_i,\, \mathcal{H}(\mathbf{z}),\, \sigma^{\mathbf{z}}\right) \stackrel{?}{=} \mathrm{true},
\end{equation}
and accepts $\mathbf{z}$ if and only if the sender's certificate is valid and unrevoked and the signature verification succeeds; otherwise the message is discarded.

This design fully decouples confidentiality from integrity: CKKS ensures that ESs cannot observe individual vehicles' plaintext parameters, while ML-DSA ensures that no exchanged data has been tampered with.

\textbf{Lagrangian Multiplier Synchronization.} During each FL aggregation round, the global aggregator updates $\lambda$ based on the average task success rate $\bar{s}$ over the past $K$ time slots reported by all ESs, and broadcasts the updated $\lambda$ together with the global Actor parameters, incurring no additional communication rounds.

\subsection{Overall Procedure and Analysis}
\label{sec:analysis} 

Algorithm~\ref{alg:decision} details the per-slot decision and training procedure described in Sections~\ref{sec:cmdp}--\ref{sec:mappo}, and Algorithm~\ref{alg:aggregation} details the periodic federated aggregation procedure described in Section~\ref{sec:fl}.

\begin{algorithm}[t]
\caption{CRAFT Decision and Training (Per Time Slot)}
\label{alg:decision}
\footnotesize
\begin{algorithmic}[1]
\REQUIRE $\omega_1, \omega_2, \theta, T_{\mathrm{batch}}, \epsilon, \gamma, \varsigma, \mu$
\STATE Initialize $\boldsymbol{w}_n\ \forall n$, $\boldsymbol{w}_m^c\ \forall m$, $\lambda \leftarrow 0$
\STATE \textcolor{gray}{\textit{// All transmissions are ML-DSA signed and verified against CA-issued certificates; messages failing verification are discarded.}}
\FOR{$t = 1, 2, \ldots, T$}
    \STATE \textcolor{gray}{\textit{// Phase 1: Action sampling}}
    \FOR{each vehicle $n \in \mathcal{N}$ \textbf{in parallel}}
        \STATE Observe $o_n(t)$
        \IF{$\mathcal{M}_n(t) = \emptyset$}
            \STATE $a_k \leftarrow 0,\ \forall k \in \mathcal{K}_n$
        \ELSE
            \STATE $a_n(t) \sim \pi_{\boldsymbol{w}_n}(\cdot \mid o_n(t))$; submit $\{a_k, \beta_k\}$ to serving ES
        \ENDIF
    \ENDFOR
    \STATE \textcolor{gray}{\textit{// Phase 2: ES-side resource allocation}}
    \FOR{each ES $m \in \mathcal{M}$}
        \STATE Allocate $f_{k,m}^{\mathrm{edge}}(t)$ via Eq.~\eqref{eq:resource_alloc}; return to vehicles
    \ENDFOR
    \STATE \textcolor{gray}{\textit{// Phase 3: Execution and reward computation}}
    \FOR{each vehicle $n$ with $\mathcal{M}_n(t) \neq \emptyset$ \textbf{in parallel}}
        \STATE Execute offloading; compute $D_n^{\mathrm{DAG}}(t)$, $E_n^{\mathrm{total}}(t)$, $s_n(t)$
        \STATE Compute augmented reward $\tilde{r}_n(t)$ via Eq.~\eqref{eq:aug_reward}
        \STATE Upload signed $\tilde{r}_n(t)$ to serving ES
    \ENDFOR
    \STATE \textcolor{gray}{\textit{// Phase 4: Critic update and advantage estimation}}
    \FOR{each ES $m \in \mathcal{M}$}
        \STATE Store $(o_n, a_n, \tilde{r}_n, o_n')$ for all $n \in \mathcal{N}_m(t)$
        \FOR{each $n$ with $T_{\mathrm{batch}}$ consecutive samples}
            \STATE Compute $\hat{A}_n^{(t)}$ via GAE~\eqref{eq:gae}; set $\hat{R}_n^{(t)} \leftarrow \hat{A}_n^{(t)} + V_{\boldsymbol{w}_m^{c,\mathrm{old}}}(s_m(t))$
            \STATE Update Critic $\boldsymbol{w}_m^c$ via Eq.~\eqref{eq:critic}; send $\hat{A}_n^{(t)}$ to vehicle $n$
        \ENDFOR
    \ENDFOR
    \STATE \textcolor{gray}{\textit{// Phase 5: Actor update}}
    \FOR{each vehicle $n$ that received $\hat{A}_n^{(t)}$}
        \STATE Update Actor $\boldsymbol{w}_n$ by minimizing $\mathcal{L}^{\mathrm{local}}(\boldsymbol{w}_n)$ in Eq.~\eqref{eq:local_loss}
    \ENDFOR
    \IF{$t \bmod K = 0$}
        \STATE Call Algorithm~\ref{alg:aggregation}
    \ENDIF
\ENDFOR
\end{algorithmic}
\end{algorithm}

\begin{algorithm}[t]
\caption{Secure Federated Aggregation (Every $K$ Slots)}
\label{alg:aggregation}
\footnotesize
\begin{algorithmic}[1]
\REQUIRE Aggregation period $K$, multiplier learning rate $\eta_\lambda$
\STATE \textcolor{gray}{\textit{// Tier 1: Intra-ES aggregation (ciphertext domain)}}
\FOR{each ES $m \in \mathcal{M}$}
    \STATE $\mathcal{N}_m^{\mathrm{FL}} \leftarrow \mathcal{N}_m(t)$
    \FOR{each $n \in \mathcal{N}_m^{\mathrm{FL}}$}
        \STATE $\boldsymbol{e}_n \leftarrow \mathrm{CKKS\_Enc}(\mathit{pk}^{\mathrm{HE}}, \boldsymbol{w}_n)$
        \STATE $\sigma_n^{\mathrm{up}} \leftarrow \mathrm{ML\text{-}DSA\_Sign}(\mathit{sk}_n, \mathcal{H}(\boldsymbol{e}_n))$; upload $(\boldsymbol{e}_n,\, \sigma_n^{\mathrm{up}})$
    \ENDFOR
    \FOR{each $n \in \mathcal{N}_m^{\mathrm{FL}}$}
        \IF{$n$'s certificate is invalid \textbf{or} ML-DSA verification fails}
            \STATE Discard $\boldsymbol{e}_n$; remove $n$ from $\mathcal{N}_m^{\mathrm{FL}}$
        \ENDIF
    \ENDFOR
    \STATE Compute $\boldsymbol{e}_m^{\mathrm{local}}$ via homomorphic aggregation~\eqref{eq:cipher_agg_tier1}
\ENDFOR
\STATE \textcolor{gray}{\textit{// Tier 2: Inter-ES aggregation}}
\STATE Global aggregator collects $\{\boldsymbol{e}_m^{\mathrm{local}}\}_{m \in \mathcal{M}}$
\STATE Compute $\boldsymbol{e}^{\mathrm{global}}$ via homomorphic aggregation~\eqref{eq:cipher_agg_tier2}
\STATE $\boldsymbol{w}^{\mathrm{global}} \leftarrow \mathrm{CKKS\_Dec}(\mathit{sk}^{\mathrm{HE}}, \boldsymbol{e}^{\mathrm{global}})$
\STATE \textcolor{gray}{\textit{// Lagrangian multiplier synchronization}}
\STATE $\bar{s} \leftarrow \frac{1}{NK}\sum_{t'=t-K+1}^{t}\sum_{n=1}^{N} s_n(t')$; update $\lambda$ via Eq.~\eqref{eq:lambda_update}
\STATE \textcolor{gray}{\textit{// ML-DSA signed broadcast for downlink integrity}}
\STATE $\sigma^{\mathrm{sig}} \leftarrow \mathrm{ML\text{-}DSA\_Sign}(\mathit{sk}_{\mathrm{agg}}, \mathcal{H}(\boldsymbol{w}^{\mathrm{global}} \| \lambda))$
\STATE Broadcast $(\boldsymbol{w}^{\mathrm{global}},\, \lambda,\, \sigma^{\mathrm{sig}})$ to all vehicles
\FOR{each recipient}
    \IF{$\mathrm{agg}$'s certificate is invalid \textbf{or} ML-DSA verification fails}
        \STATE Discard message; retain previous $\boldsymbol{w}^{\mathrm{global}}$ and $\lambda$
    \ENDIF
\ENDFOR
\end{algorithmic}
\end{algorithm}

\textbf{Convergence and Stability Analysis.}
CRAFT trains on two interleaved timescales that are intrinsic to its design. On the fast timescale, each vehicle performs, every slot, a proximally regularized policy update on the augmented reward $\tilde r=r-\lambda c$. On the slow timescale, every $K$ slots the dual variable $\lambda$ ascends on the empirical completion-rate violation and the local Actors are federated. The multiplier therefore varies slowly relative to the policy and defines a constrained-optimal \emph{target} that the federated Actors must track, while the proximal coefficient $\mu$ governs both the tightness of that tracking and the damping of non-IID dispersion. Our analysis follows this structure: Proposition~\ref{prop:constraint} expresses the average constraint violation as a dual-schedule part plus a single primal tracking residual; Proposition~\ref{prop:policy} bounds how fast the slow timescale moves the target; and Proposition~\ref{prop:federated} bounds the fast-timescale tracking error, where the affine-in-$\lambda$ structure of $\tilde r$ makes one coefficient $\mu$ control two distinct error sources. As CRAFT's policies are non-convex parametrizations optimized inexactly over stochastic, non-IID experience, the results are conditional stability bounds with an explicit optimization floor, not unconditional global convergence.

\begin{proposition}[Dual-Tracking Constraint Bound]
\label{prop:constraint}
Let $g^{(j)}=\theta-\bar{s}^{(j)}$ be the empirical completion-rate violation at the $j$-th multiplier-update round, $j=1,\ldots,J$, with $g^{(j)}\le 0$ meaning the constraint is met in that round's average. Since $c_n(t)=\theta-s_n(t)\in[\theta-1,\theta]$, we have $|g^{(j)}|\le C_c$ with $C_c=\max\{\theta,1-\theta\}$. Under the projected dual update $\lambda^{(j+1)}=\Pi_{[0,\lambda^{\max}]}\!\big(\lambda^{(j)}+\eta_\lambda g^{(j)}\big)$,
\begin{equation}
\resizebox{0.89\columnwidth}{!}{$\left[\frac{1}{J}\sum_{j=1}^{J}g^{(j)}\right]_+ \le \underbrace{\frac{\bar\rho}{\lambda^{\max}}}_{\text{primal tracking}} + \underbrace{\frac{\lambda^{\max}}{2\eta_\lambda J}+\frac{\eta_\lambda C_c^2}{2\lambda^{\max}}}_{\text{dual schedule}}, \quad \bar\rho:=\frac{1}{J}\Big[\sum_{j=1}^{J}\lambda^{(j)}g^{(j)}\Big]_+ .$}
\end{equation}
The dual-schedule terms vanish at rate $O(1/\sqrt{J})$ under $\eta_\lambda=\Theta(1/\sqrt{J})$. The remaining term $\bar\rho$ is the only primal-dependent quantity: it accumulates the penalty actually paid while the constraint is still violated, and is small precisely when the policy restores feasibility before $\lambda$ grows large, i.e.\ when the fast timescale tracks the slow one. We bound it through Propositions~\ref{prop:policy}--\ref{prop:federated} rather than assuming its order.
\end{proposition}

\begin{proof}
By non-expansiveness of $\Pi_{[0,\lambda^{\max}]}$, for the comparator $\lambda^{\max}$, $\|\lambda^{(j+1)}-\lambda^{\max}\|^2\le\|\lambda^{(j)}+\eta_\lambda g^{(j)}-\lambda^{\max}\|^2$. Expanding, rearranging, summing over $j$, and using $|g^{(j)}|\le C_c$ together with $\|\lambda^{(1)}-\lambda^{\max}\|^2\le(\lambda^{\max})^2$ gives
\begin{equation}
\resizebox{0.89\columnwidth}{!}{$\lambda^{\max}\sum_{j=1}^{J}g^{(j)} \le \sum_{j=1}^{J}\lambda^{(j)}g^{(j)} +\frac{(\lambda^{\max})^2}{2\eta_\lambda} +\frac{\eta_\lambda J C_c^2}{2}.$}
\end{equation}
Dividing by $\lambda^{\max}J$ and taking the positive part, and noting $\bar\rho\le\tfrac1J\sum_j\lambda^{(j)}[g^{(j)}]_+$, yields the stated decomposition.
\end{proof}

\begin{proposition}[Slow-Timescale Target Drift]
\label{prop:policy}
For fixed $\lambda$ the augmented return $\tilde r_\lambda=r-\lambda c$ induces a stationary objective $J_\lambda$. Assume $|r(t)|\le R_{\max}$, $|c(t)|\le C_c$, and that the clipped PPO update keeps the empirical KL divergence between consecutive policies below $\delta_{\mathrm{KL}}$. Then within a $\lambda$-fixed interval
\begin{equation}
\big|J_\lambda(\pi^{(i+1)})-J_\lambda(\pi^{(i)})\big| \le C_{\mathrm{KL}}\sqrt{\delta_{\mathrm{KL}}}+\epsilon_{\mathrm{est}},
\end{equation}
with $C_{\mathrm{KL}}$ depending on $R_{\max}+\lambda^{\max}C_c$ and $\gamma$, and $\epsilon_{\mathrm{est}}$ collecting sampling, critic, and advantage-estimation error. Across a slow-timescale step, the projection gives $|\Delta\lambda^{(j)}|:=|\lambda^{(j+1)}-\lambda^{(j)}|\le\eta_\lambda|g^{(j)}|\le\eta_\lambda C_c$, so the value of any fixed policy $\pi$ moves by at most
\begin{equation}
\big|V^\pi_{\lambda^{(j+1)}}(s)-V^\pi_{\lambda^{(j)}}(s)\big| \le\frac{|\Delta\lambda^{(j)}|\,C_c}{1-\gamma} \le\frac{\eta_\lambda C_c^2}{1-\gamma}.
\end{equation}
Hence the constrained-optimal target drifts at per-round rate $O(\eta_\lambda)$, controllable through the dual stepsize.
\end{proposition}

\begin{proof}
For the first bound, $J_\lambda$ is stationary for fixed $\lambda$ and its rewards are bounded by $R_{\max}+\lambda^{\max}C_c$; Pinsker's inequality bounds the total-variation distance between consecutive policies by $\sqrt{\delta_{\mathrm{KL}}/2}$, and the performance-difference lemma converts this into an objective gap of order $\sqrt{\delta_{\mathrm{KL}}}$ up to $\epsilon_{\mathrm{est}}$. For the second, $|\Delta\lambda^{(j)}|\le\eta_\lambda|g^{(j)}|$ by non-expansiveness, and for any fixed $\pi$,
\begin{equation}
\resizebox{0.89\columnwidth}{!}{$\big|V^\pi_{\lambda^{(j+1)}}(s)-V^\pi_{\lambda^{(j)}}(s)\big| =\Big|\mathbb{E}_\pi\!\sum_{\ell=0}^{\infty}\gamma^\ell(\lambda^{(j)}-\lambda^{(j+1)})c(t+\ell)\Big| \le\frac{|\Delta\lambda^{(j)}|C_c}{1-\gamma}.$}
\end{equation}
\end{proof}

\begin{proposition}[Dual-Role Proximal Tracking Bound]
\label{prop:federated}
Because $\tilde r=r-\lambda c$ is affine in $\lambda$ and GAE is a $(\gamma\varsigma)$-weighted sum of one-step temporal-difference errors that is linear in the per-step reward, the augmented advantage entering the surrogate splits as a reward part plus a $\lambda$-scaled cost part. Consequently the clipped-surrogate stochastic gradient obeys
\begin{equation}
\big\|\nabla\mathcal{L}_n^{\mathrm{clip}}(\boldsymbol{w};\lambda,\xi)\big\|\le G_r+\lambda\,G_c,
\end{equation}
where $G_r$ and $G_c$ bound the reward- and cost-induced gradient magnitudes. Let the round-$p$ local objective be $\mathcal{L}_n^{\mathrm{local}}(\boldsymbol{w})=\mathcal{L}_n^{\mathrm{clip}}(\boldsymbol{w};\lambda^{(p)})+\tfrac{\mu}{2}\|\boldsymbol{w}-\boldsymbol{w}^{(p)}\|^2$. If each vehicle takes $K$ local steps with stepsize $\eta_w$ and $0<\eta_w\mu<1$,
\begin{equation}
\resizebox{0.89\columnwidth}{!}{$\mathbb{E}\big\|\boldsymbol{w}_n^{(p,K)}-\boldsymbol{w}^{(p)}\big\| \le\frac{(G_r+\lambda^{(p)}G_c)\big(1-(1-\eta_w\mu)^K\big)}{\mu} \le\underbrace{\frac{G_r}{\mu}}_{\text{non-IID}}+\underbrace{\frac{\lambda^{(p)}G_c}{\mu}}_{\lambda\text{-shock}},$}
\end{equation}
and the change in drift attributable to a multiplier update is at most $G_c|\Delta\lambda^{(p)}|/\mu\le G_c\eta_\lambda C_c/\mu$. A single coefficient $\mu$ thus damps both heterogeneity-driven dispersion ($G_r/\mu$) and dual-induced oscillation ($\lambda^{(p)}G_c/\mu$), which is the formal basis of the proximal term's dual role.
\end{proposition}

\begin{proof}
With $\tilde r=r-\lambda c$, the GAE return decomposes by linearity into a reward return minus $\lambda$ times a cost return; the clip threshold depends only on the importance ratio, so on any active branch the surrogate gradient is proportional to this advantage, and the triangle inequality gives $\|\nabla\mathcal{L}_n^{\mathrm{clip}}(\boldsymbol{w};\lambda,\xi)\|\le G_r+\lambda G_c$, with critic-approximation error absorbed into $G_r,G_c$. Writing $\boldsymbol{d}_k=\boldsymbol{w}_n^{(p,k)}-\boldsymbol{w}^{(p)}$, a proximal step gives $\boldsymbol{d}_{k+1}=(1-\eta_w\mu)\boldsymbol{d}_k-\eta_w\nabla\mathcal{L}_n^{\mathrm{clip}}(\boldsymbol{w}_n^{(p,k)};\lambda^{(p)},\xi)$. Taking norms and expectations, using the gradient bound and $\boldsymbol{d}_0=0$,
\begin{equation}
\resizebox{0.89\columnwidth}{!}{$\mathbb{E}\|\boldsymbol{d}_K\| \le\eta_w(G_r+\lambda^{(p)}G_c)\sum_{q=0}^{K-1}(1-\eta_w\mu)^q =\frac{(G_r+\lambda^{(p)}G_c)\big(1-(1-\eta_w\mu)^K\big)}{\mu}.$}
\end{equation}
Differencing the $\lambda^{(p)}G_c/\mu$ term across rounds with $|\Delta\lambda^{(p)}|\le\eta_\lambda C_c$ gives the shock-damping bound.
\end{proof}

\textbf{Synthesis.} The three results compose along the two timescales. By Proposition~\ref{prop:policy} the slow timescale moves the constrained-optimal target at per-round rate $O(\eta_\lambda C_c^2/(1-\gamma))$; by Proposition~\ref{prop:federated} each federated round contracts the local drift toward the target while absorbing the multiplier shock with magnitude $O(G_c\eta_\lambda C_c/\mu)$. The per-round primal tracking error is therefore $O(\eta_\lambda)+\rho_0$, where $\rho_0\ge 0$ is an irreducible floor set by non-convex inexact policy optimization, critic approximation, and bounded function approximation; the underlying CMDP has zero duality gap in occupancy-measure space~\cite{altman1999cmdp}, so $\rho_0$ reflects an optimization gap, not a duality gap. This tracking error governs the residual $\bar\rho$ of Proposition~\ref{prop:constraint}, giving an average positive constraint violation of $O(1/\sqrt{J})+\rho_0/\lambda^{\max}$: the dual-controllable part decays at the $O(1/\sqrt{J})$ rate, and the residual floor is shaped jointly by the dual stepsize $\eta_\lambda$ and the proximal coefficient $\mu$, a coupling that arises only under CRAFT's combined constrained, federated, augmented-reward design. Empirical training stability and convergence are examined in Section~\ref{sec:experiments}.

\textbf{Complexity Analysis.} The per-slot computational complexity of CRAFT is dominated by the forward and backward passes of the Actor and Critic networks. Let $|\boldsymbol{w}|$ and $|\boldsymbol{w}^c|$ denote the number of parameters in the Actor and Critic networks, respectively, and $h$ the embedding dimension. The per-vehicle per-slot complexity is $O(K_n^2 h + |\boldsymbol{w}|)$, where the first term accounts for self-attention over subtask embeddings. The per-ES per-slot complexity is $O(|\mathcal{N}_m(t)|^2 h + |\mathcal{N}_m(t)| \cdot |\boldsymbol{w}^c|)$, where the first term accounts for self-attention over vehicle embeddings. FL aggregation is performed every $K$ time slots, with intra-ES aggregation complexity $O(|\mathcal{N}_m^{\mathrm{FL}}| \cdot |\boldsymbol{w}|)$ and inter-ES aggregation complexity $O(M \cdot |\boldsymbol{w}|)$, both linear. CKKS encryption and decryption scale linearly with the number of Actor parameters, requiring $O(|\boldsymbol{w}|)$ encryption/decryption operations. The ML-DSA signing and verification complexity is $O(1)$ with respect to payload size, since both operations act on a fixed-size hash digest.

\textbf{Communication Overhead Analysis.} All messages transmitted between entities carry an ML-DSA digital signature (2,420--4,627 bytes depending on the security level~\cite{fips204}), which is small relative to task offloading data volumes. The communication overhead in the decision loop is minimal: each vehicle uploads a single scalar ($\tilde{r}_n(t)$) per slot and the ES returns a single scalar (the advantage value). In the aggregation loop, each vehicle uploads encrypted Actor parameters with a communication volume of $\zeta \cdot |\boldsymbol{w}|$, where $\zeta > 1$ denotes the CKKS ciphertext expansion factor, i.e., the ratio of ciphertext size to plaintext size, and the inter-ES backhaul transmission overhead is $M \cdot \zeta \cdot |\boldsymbol{w}|$. Since FL aggregation occurs only once every $K$ time slots, its communication overhead is manageable compared to the data transmission volume of task offloading itself.

\section{Performance Evaluation}
\label{sec:experiments}

This section evaluates CRAFT's performance through extensive experiments on a VM-based platform with realistic vehicular workloads.

\subsection{Experimental Setup}

\textbf{Experimental Platform.} We implement and evaluate CRAFT on a VM-based platform, where each vehicle and ES runs on a dedicated VM. To reflect heterogeneous real-world vehicle computing, vehicle VMs are configured with reference to mainstream intelligent vehicles (e.g., Tesla Model 3\footnote{\url{https://www.tesla.com/ownersmanual/model3/en_us/GUID-2CB60804-9CEA-4F4B-8B04-09B991368DC5.html}}, Volvo EX90\footnote{\url{https://www.nvidia.com/content/dam/en-zz/Solutions/gtcf21/jetson-orin/nvidia-jetson-agx-orin-technical-brief.pdf}}, BMW iX3\footnote{\url{https://www.qualcomm.com/automotive/solutions/snapdragon-ride}}), each with 4--12 cores and 4--32~GB RAM. Each ES VM has 16 cores and 64~GB RAM, matching servers on roadside units. Inter-VM delays follow the channel model in Section~\ref{sec:system_model}.

\begin{figure}[t]
    \centering
    \includegraphics[width=0.8\columnwidth]{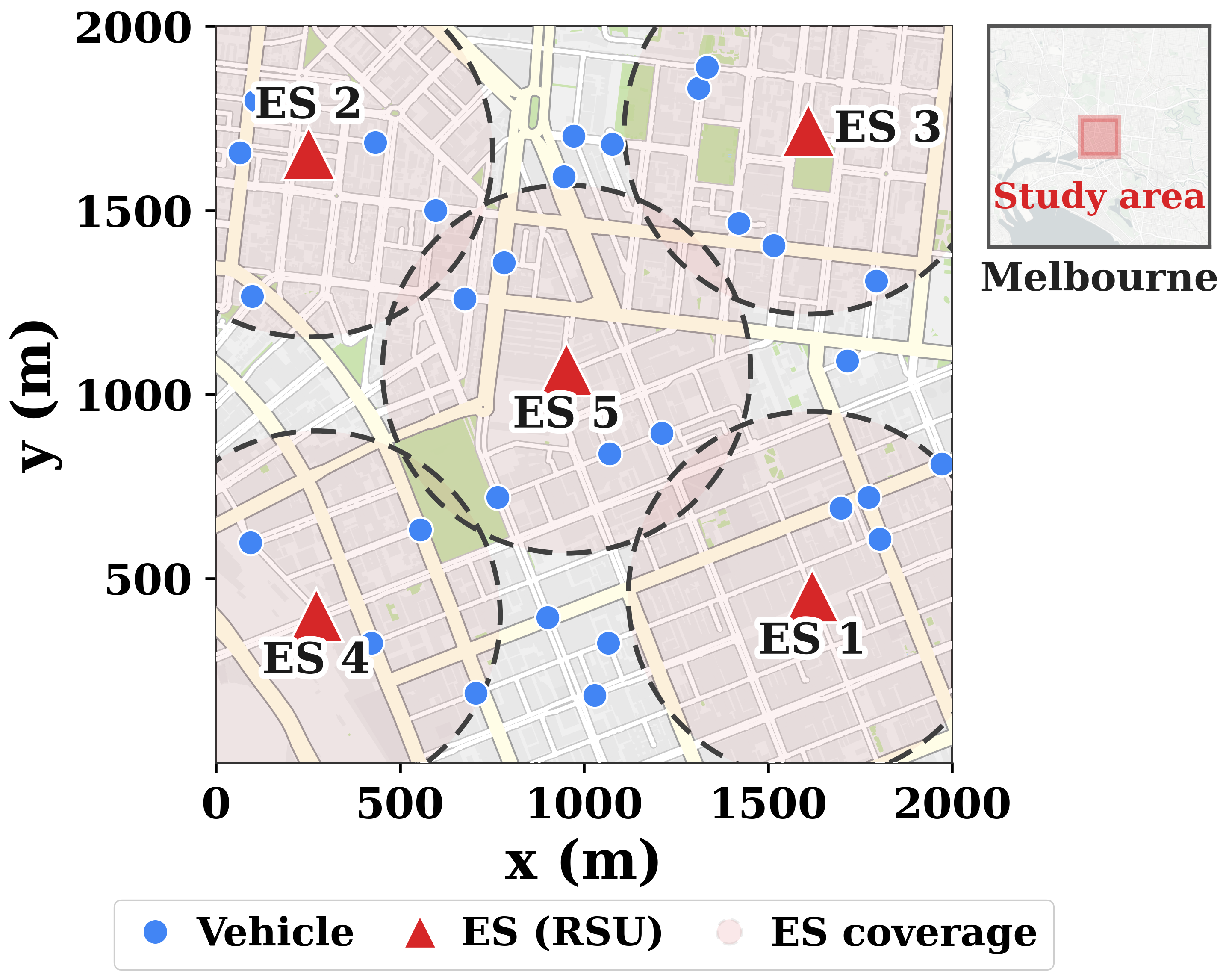}
    \caption{Road scenario extracted from the Melbourne CBD via OpenStreetMap, with 5 ESs and 30 vehicles.}
    \label{fig:exp}
\end{figure}

\textbf{Road Scenario.} Vehicle trajectories are generated by SUMO~\cite{sumo} over a realistic $\sim$2~km~$\times$~2~km road network extracted from the Melbourne CBD via OpenStreetMap\footnote{\url{https://www.openstreetmap.org}}, with $M = 5$ ESs and $N = 30$ vehicles by default (Fig.~\ref{fig:exp}). Vehicle speeds range from 30 to 50~km/h per the local speed regulations\footnote{\url{https://www.melbourne.vic.gov.au/road-safety-and-speed-limits}}. The slot duration $\tau = 1$~s limits displacement to $\le$14~m per slot ($<$3\% of the coverage radius), supporting the quasi-static channel assumption.

\textbf{Workload Design.} We develop three representative vehicular DAG tasks~\cite{gong2023edge_its, katare2023edge_ad}: (i) autonomous driving perception~\cite{chen2024e2ead}, (ii) augmented reality navigation~\cite{liu2019edge_ar}, and (iii) in-vehicle voice assistant~\cite{prabhavalkar2024e2easr}. Each comprises interdependent subtasks run as real computation kernels (matrix operations, Fast Fourier Transform (FFT), image processing, compression, etc.) on the VMs. Heterogeneity arises from parameterized inputs: same-type subtasks receive varying input sizes across slots (e.g., different image resolutions or audio durations), naturally varying execution time and output volume. All delays and energy values are measured on the VMs.

\textbf{Communication Parameters.} The system bandwidth is $B = 20$~MHz, vehicle transmit power $P_n^{\mathrm{tx}} = 0.3$~W, ES transmit power $P_m^{\mathrm{ES}} = 2$~W, path loss exponent $\alpha = 3$, reference channel gain $h_0 = -30$~dB, noise power spectral density $N_0 = -174$~dBm/Hz, ES coverage radius $R_m = 500$~m, and backhaul transmission rate $r^{\mathrm{wire}} = 1$~Gbps. These parameters follow 3rd Generation Partnership Project Technical Report (3GPP TR) 36.885~\cite{3gpp_v2x} and existing VEC literature~\cite{10779474, 11085134, 11096732}.

\textbf{Baselines.} We compare CRAFT against the following three state-of-the-art methods:

\begin{itemize}
    \item \textbf{DPTORA}~\cite{10779474}: A distributed task offloading and resource allocation scheme based on the Generalist Pursuit Learning Algorithm (GPLA).
    \item \textbf{Fed-MADDPG}~\cite{11085134}: A multi-agent scheme based on MADDPG combined with FL for cross-ES knowledge sharing.
    \item \textbf{LF-MAPPO}~\cite{11096732}: A multi-agent scheme based on MAPPO with FL for inter-vehicle collaborative learning.
\end{itemize}

\textbf{Hyperparameter Tuning.} For fair comparison, the key hyperparameters of CRAFT and all baselines are tuned via grid search. In CRAFT, both the Actor and Critic adopt fully connected embedding layers ($h=256$) with single-head self-attention and a learnable aggregation token to handle variable-length inputs. CKKS is implemented with the TenSEAL library\footnote{\url{https://github.com/OpenMined/TenSEAL}} (ring dimension $N_{\mathrm{poly}}=8192$, scale $2^{40}$), and the ML-DSA-44 signatures (National Institute of Standards and Technology (NIST) security level~2) with the liboqs library\footnote{\url{https://github.com/open-quantum-safe/liboqs}}. Full settings are provided in Appendix Table~A2; each experiment is repeated 10 times and averaged.

\subsection{Convergence Performance}

\begin{figure*}[t]
    \centering
    \includegraphics[width=0.55\textwidth, height=0.6cm]{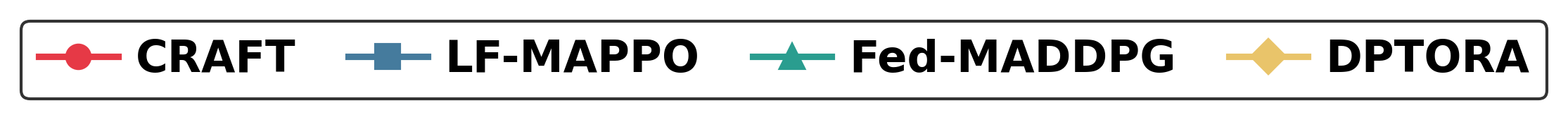}\\[-1.5pt]
    \begin{subfigure}[b]{0.24\textwidth}
        \includegraphics[width=\textwidth, height=2.5cm]{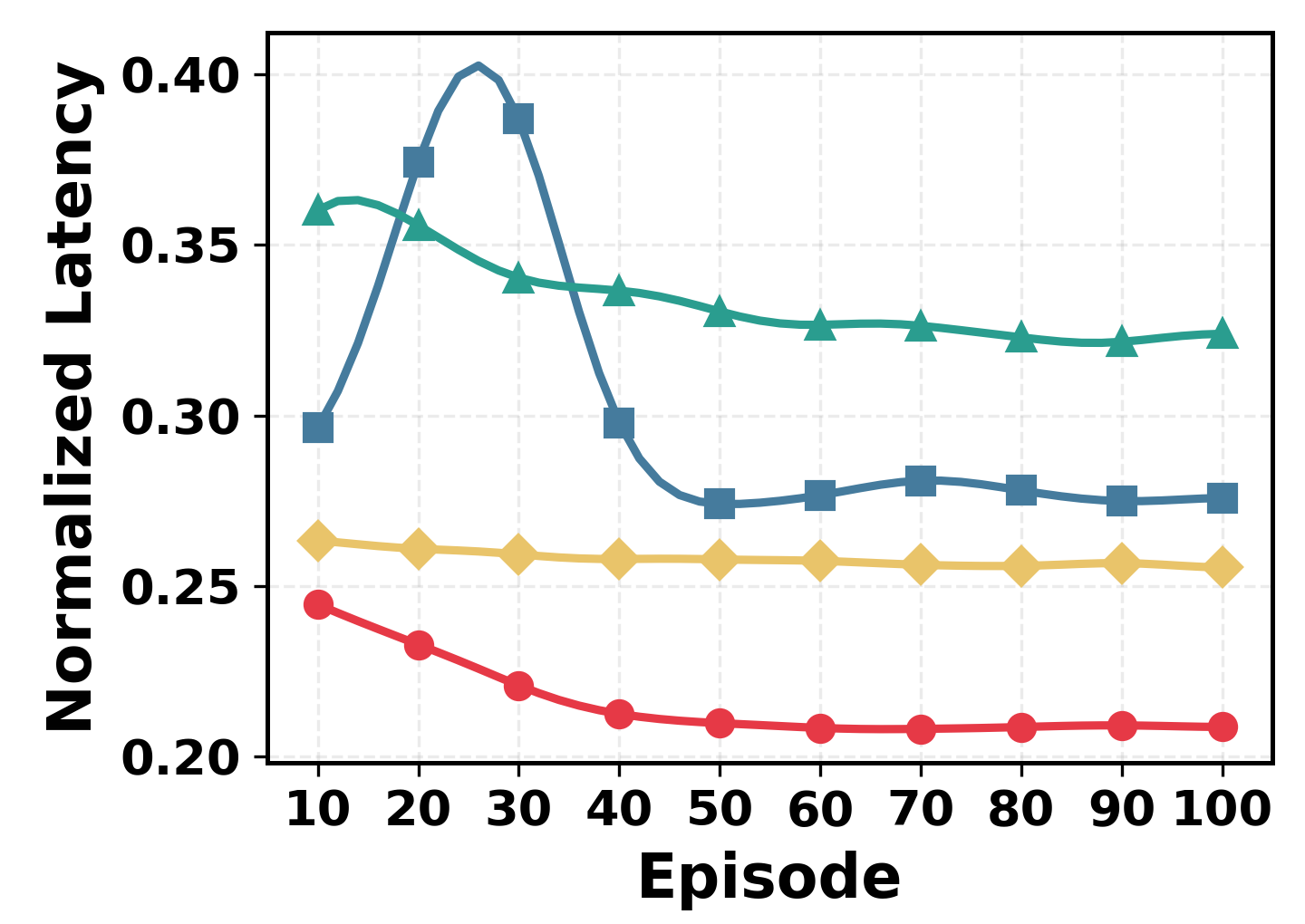}
        \caption{Normalized Latency}\label{fig:conv_latency}
    \end{subfigure}\hfill
    \begin{subfigure}[b]{0.24\textwidth}
        \includegraphics[width=\textwidth, height=2.5cm]{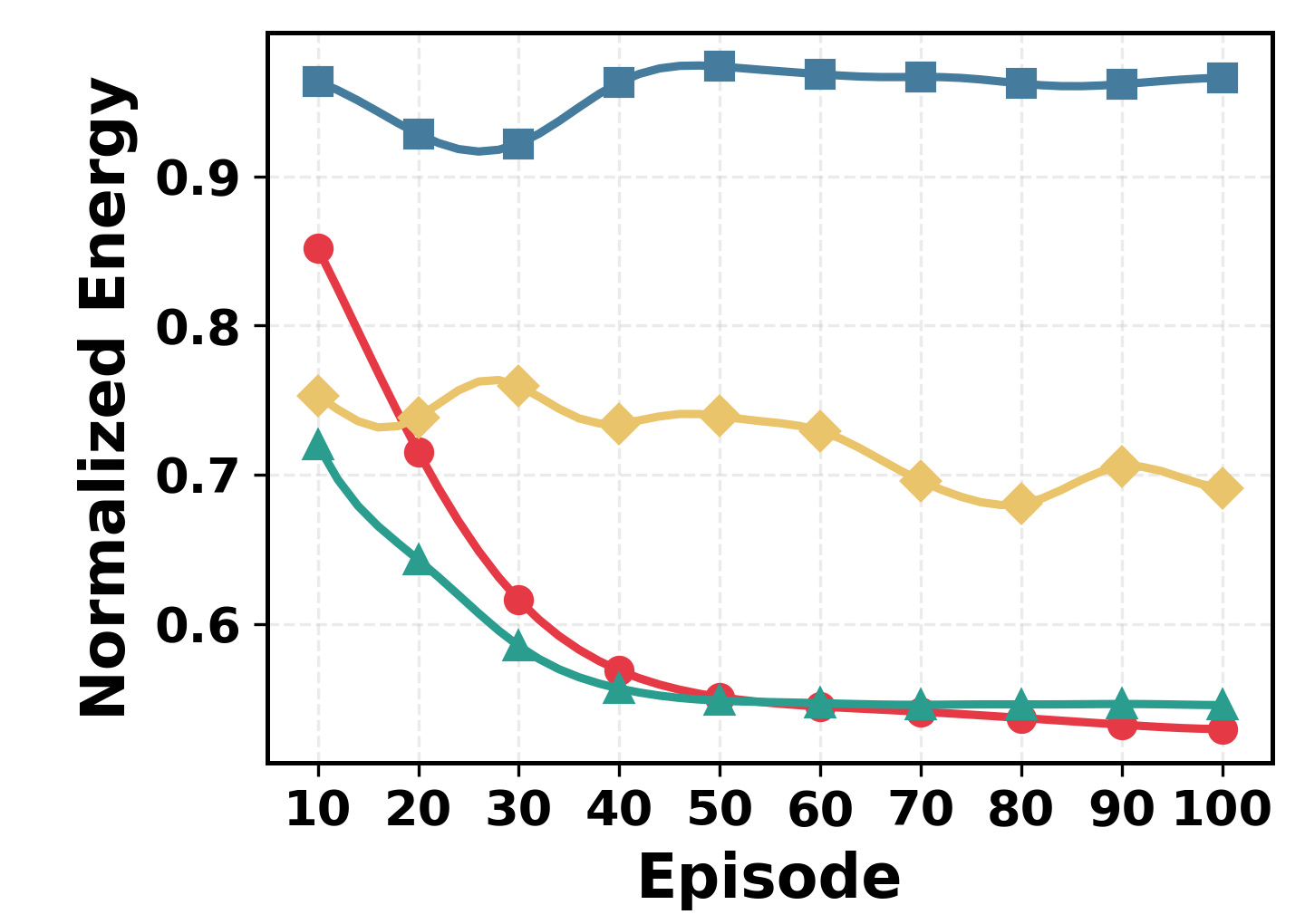}
        \caption{Normalized Energy}\label{fig:conv_energy}
    \end{subfigure}\hfill
    \begin{subfigure}[b]{0.24\textwidth}
        \includegraphics[width=\textwidth, height=2.5cm]{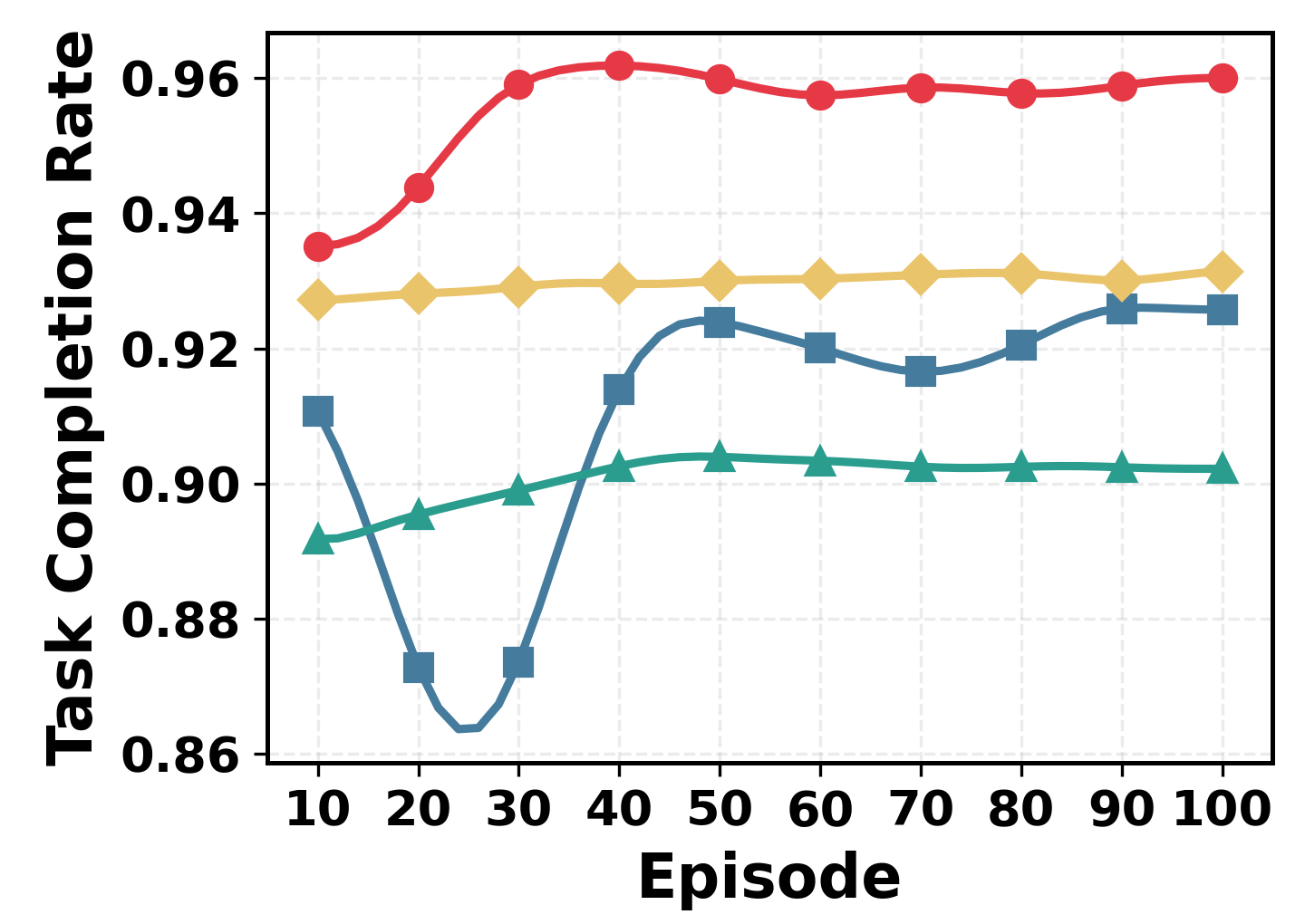}
        \caption{Task Completion Rate}\label{fig:conv_cr}
    \end{subfigure}\hfill
    \begin{subfigure}[b]{0.24\textwidth}
        \includegraphics[width=\textwidth, height=2.5cm]{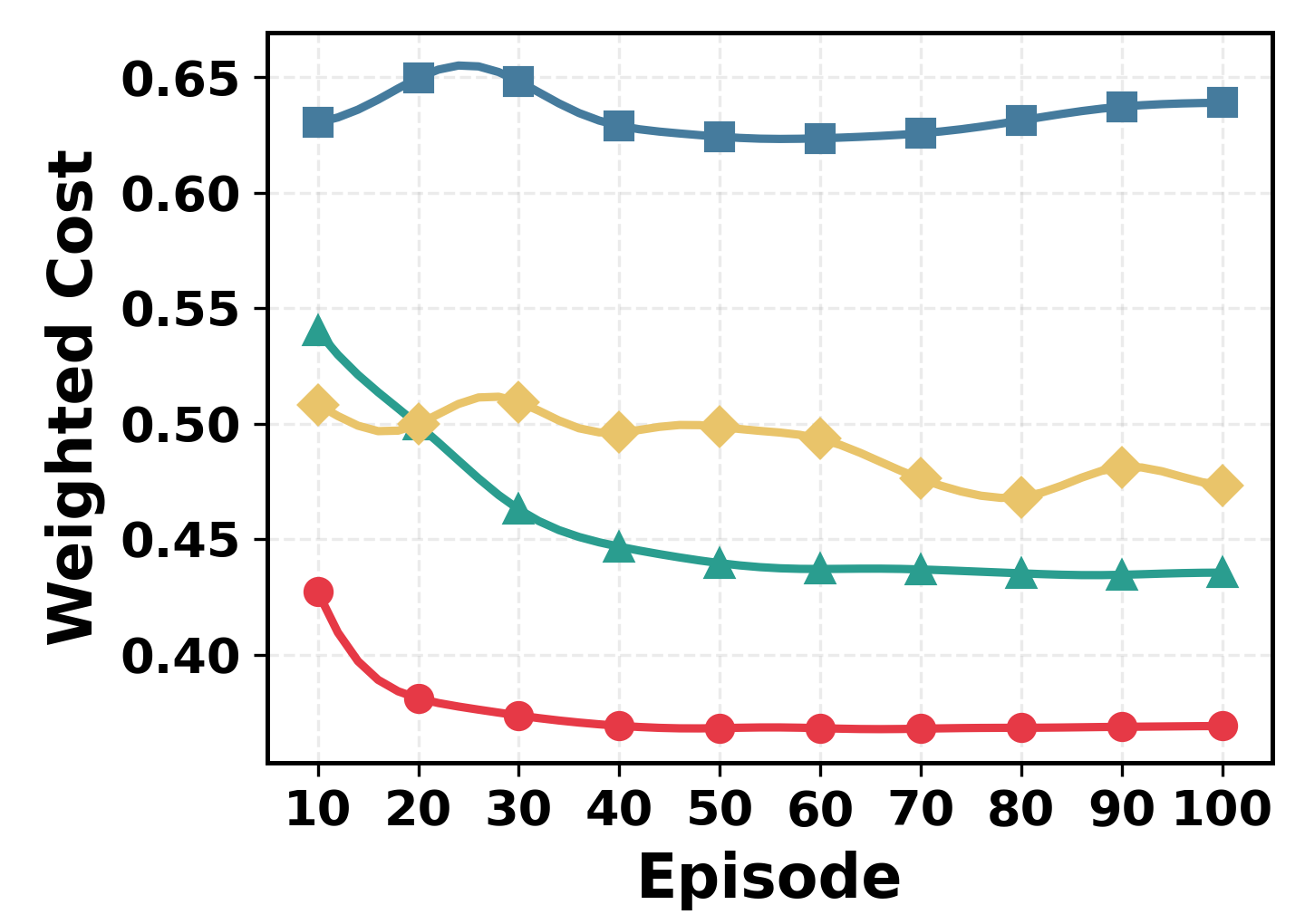}
        \caption{Weighted Cost}\label{fig:conv_cost}
    \end{subfigure}
    \caption{Convergence performance of all methods.}
    \label{fig:convergence}
\end{figure*}

Fig.~\ref{fig:convergence} shows the convergence of all methods over 100 episodes in normalized latency, energy, task completion rate, and weighted cost. CRAFT converges fastest, with all metrics stabilizing around episode 50 at the lowest values of 0.21 (latency), 0.52 (energy), and 0.37 (cost). Its completion rate settles at about 0.96, closely tracking the threshold $\theta = 0.95$ and validating the online Lagrangian adjustment.

The baselines each reveal a structural limitation. LF-MAPPO is unstable early on (latency 0.40, completion rate 0.86) under unconstrained exploration; although its latency drops to 0.27, its energy never improves, as it favors local execution to avoid communication delay. Fed-MADDPG attains competitive energy via federated strategy sharing but converges to a high latency of 0.32, limited by the poor exploration of deterministic policy gradients in the hybrid action space. DPTORA, driven by learning-automata updates, shows small fluctuations but weak optimization across all metrics, reflecting insufficient policy expressiveness. Overall, CRAFT reduces weighted cost by about 16\%, 20\%, and 42\% over Fed-MADDPG, DPTORA, and LF-MAPPO, respectively.

\subsection{Ablation Study}

\begin{figure}[t]
    \centering
    \includegraphics[width=0.7\columnwidth, height=0.4cm]{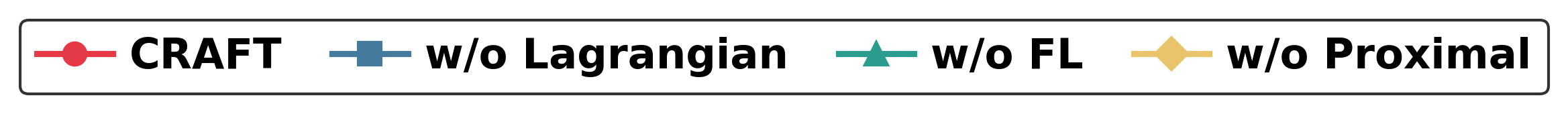}\\[-1.5pt]
    \begin{subfigure}[b]{0.48\columnwidth}
        \includegraphics[width=\textwidth, height=2.5cm]{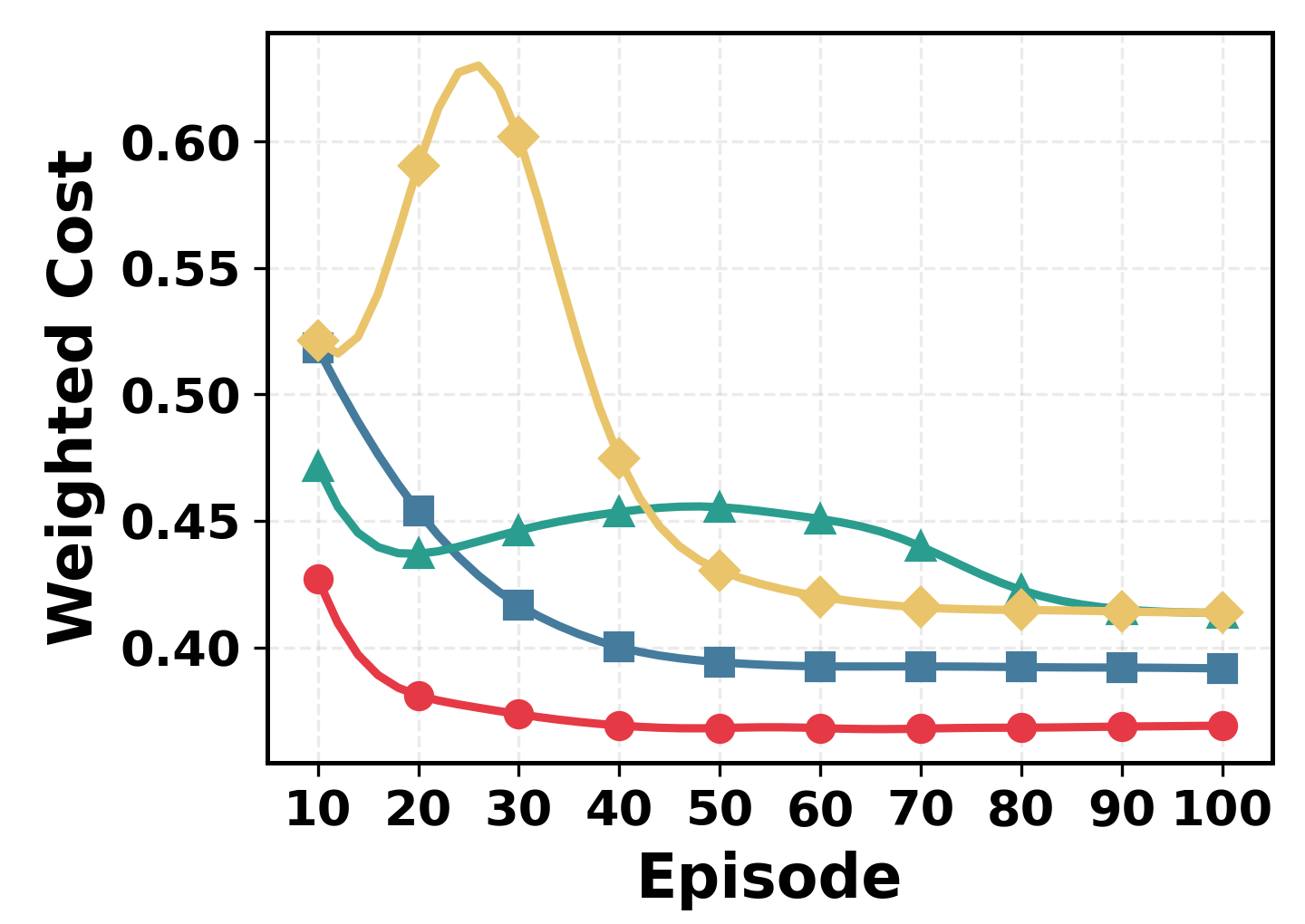}
        \caption{Weighted Cost}
        \label{fig:ablation_cost}
    \end{subfigure}
    \hfill
    \begin{subfigure}[b]{0.48\columnwidth}
        \includegraphics[width=\textwidth, height=2.5cm]{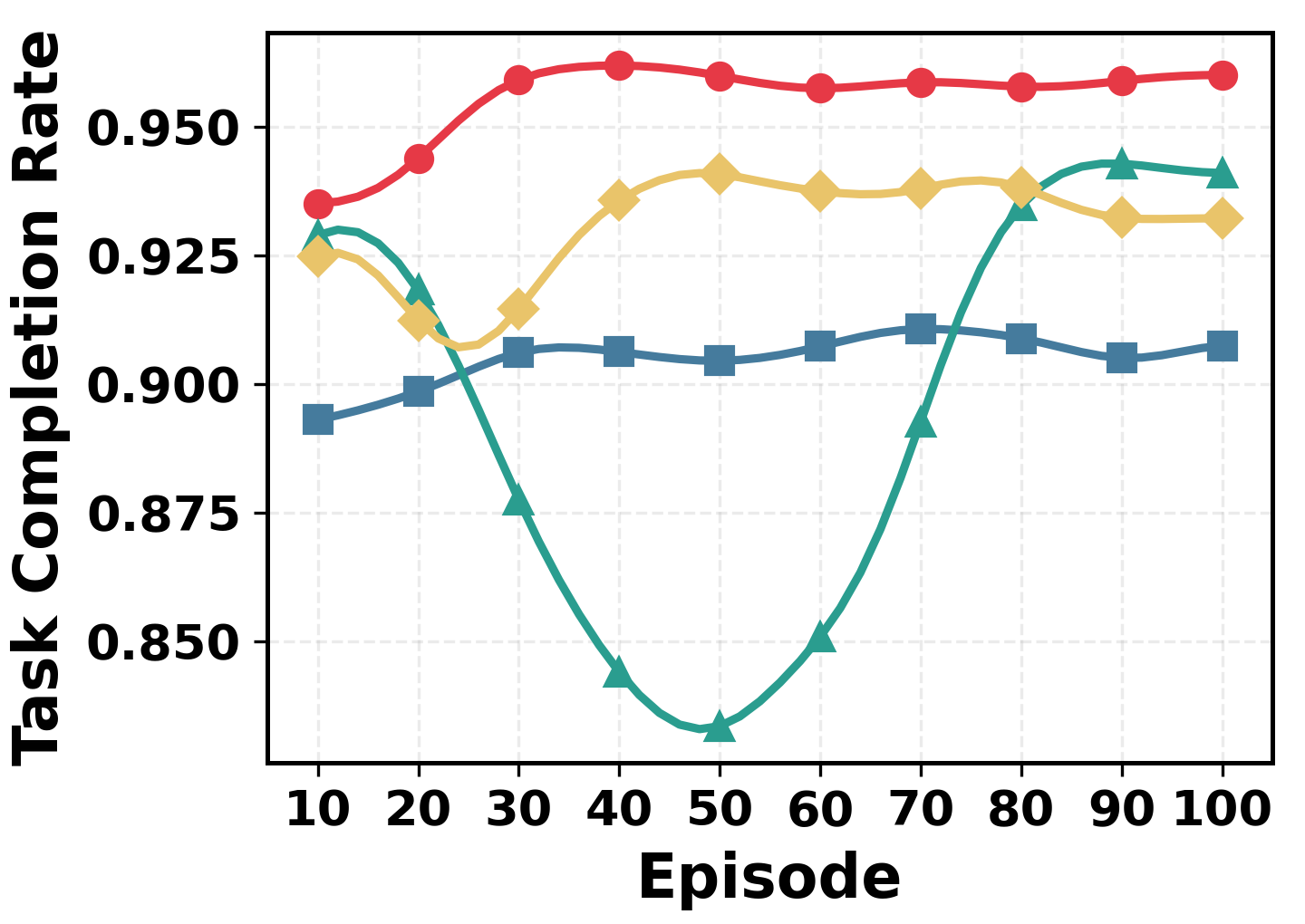}
        \caption{Task Completion Rate}
        \label{fig:ablation_cr}
    \end{subfigure}
    \caption{Ablation study on the core components of CRAFT.}
    \label{fig:ablation}
\end{figure}

Fig.~\ref{fig:ablation} isolates each core component by removing the Lagrangian mechanism, federated aggregation, and proximal regularization, respectively. Removing the Lagrangian mechanism (w/o Lagrangian) raises the weighted cost only mildly to 0.39, but the completion rate settles at just 0.91: without the dynamic penalty, the policy drifts toward lower latency and energy while neglecting task success, confirming the multiplier adjustment is indispensable for constraint satisfaction. Removing federated aggregation (w/o FL) forces each vehicle to train on local experience alone; its completion rate plunges to 0.83 mid-training and only fluctuates back to 0.94, while the cost rises to 0.41, since region-specific local optima cause severe policy mismatch at ES boundaries. Removing proximal regularization (w/o Proximal) reduces aggregation to plain FedAvg and yields the most violent early-training fluctuations, peaking at 0.63: under the strongly non-IID data across regions, unconstrained local drift destabilizes the aggregated model, and although the cost eventually reaches 0.41, the completion rate stays at only 0.93.

\subsection{Scalability Analysis}

\begin{figure}[t]
    \centering
    \includegraphics[width=0.7\columnwidth, height=0.4cm]{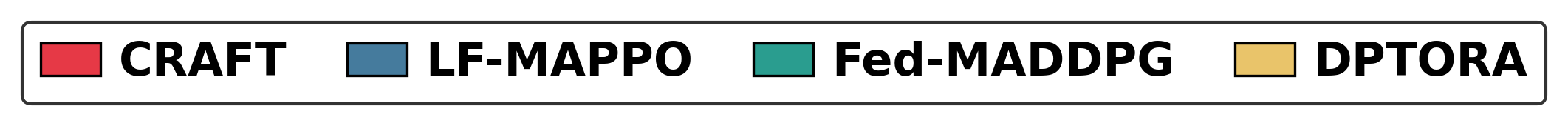}\\[-1.5pt]
    \begin{subfigure}[t]{0.48\columnwidth}
        \includegraphics[width=\textwidth, height=2.5cm]{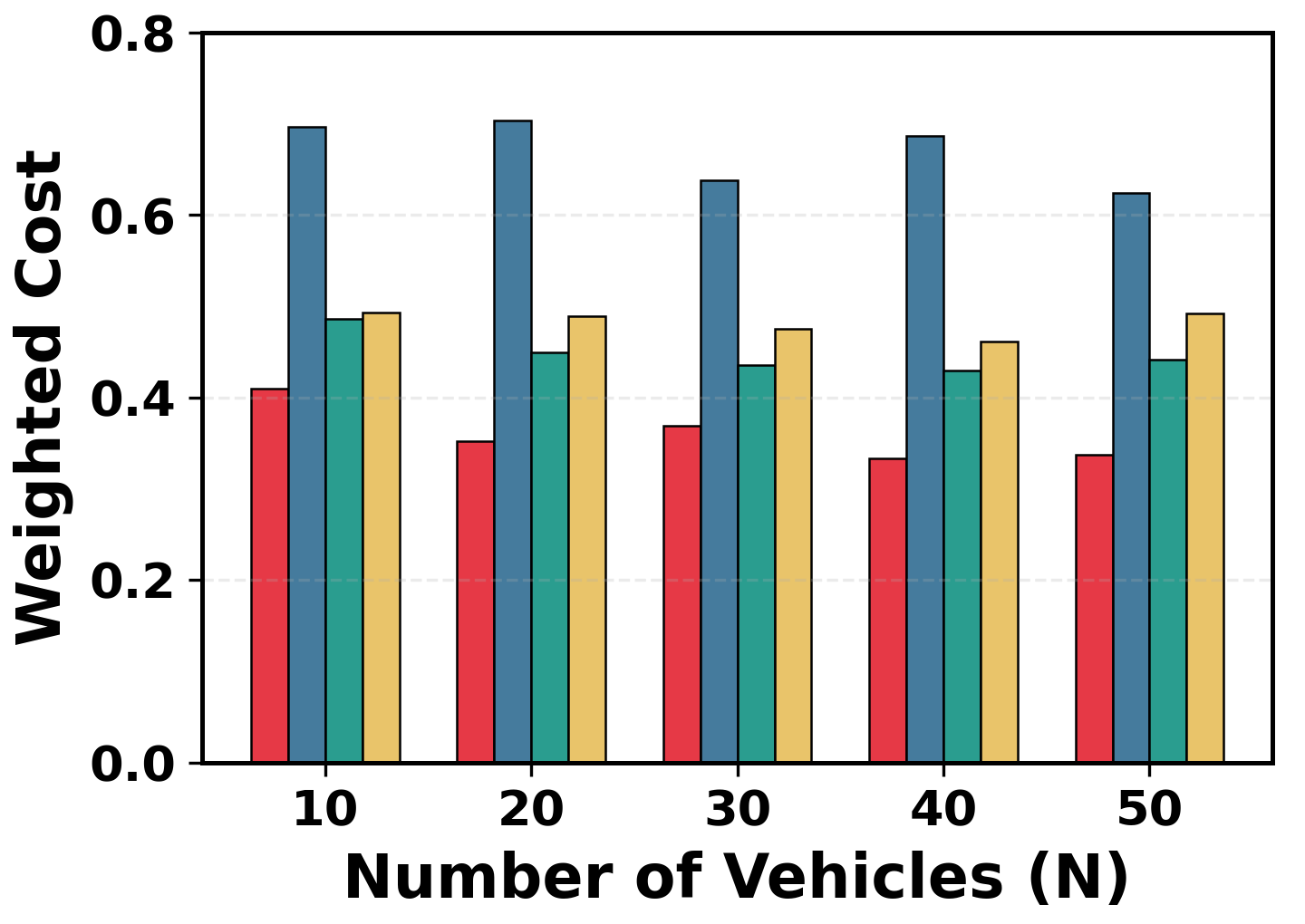}
        \caption{Varying number of vehicles}
        \label{fig:vary_n}
    \end{subfigure}
    \hfill
    \begin{subfigure}[t]{0.48\columnwidth}
        \includegraphics[width=\textwidth, height=2.5cm]{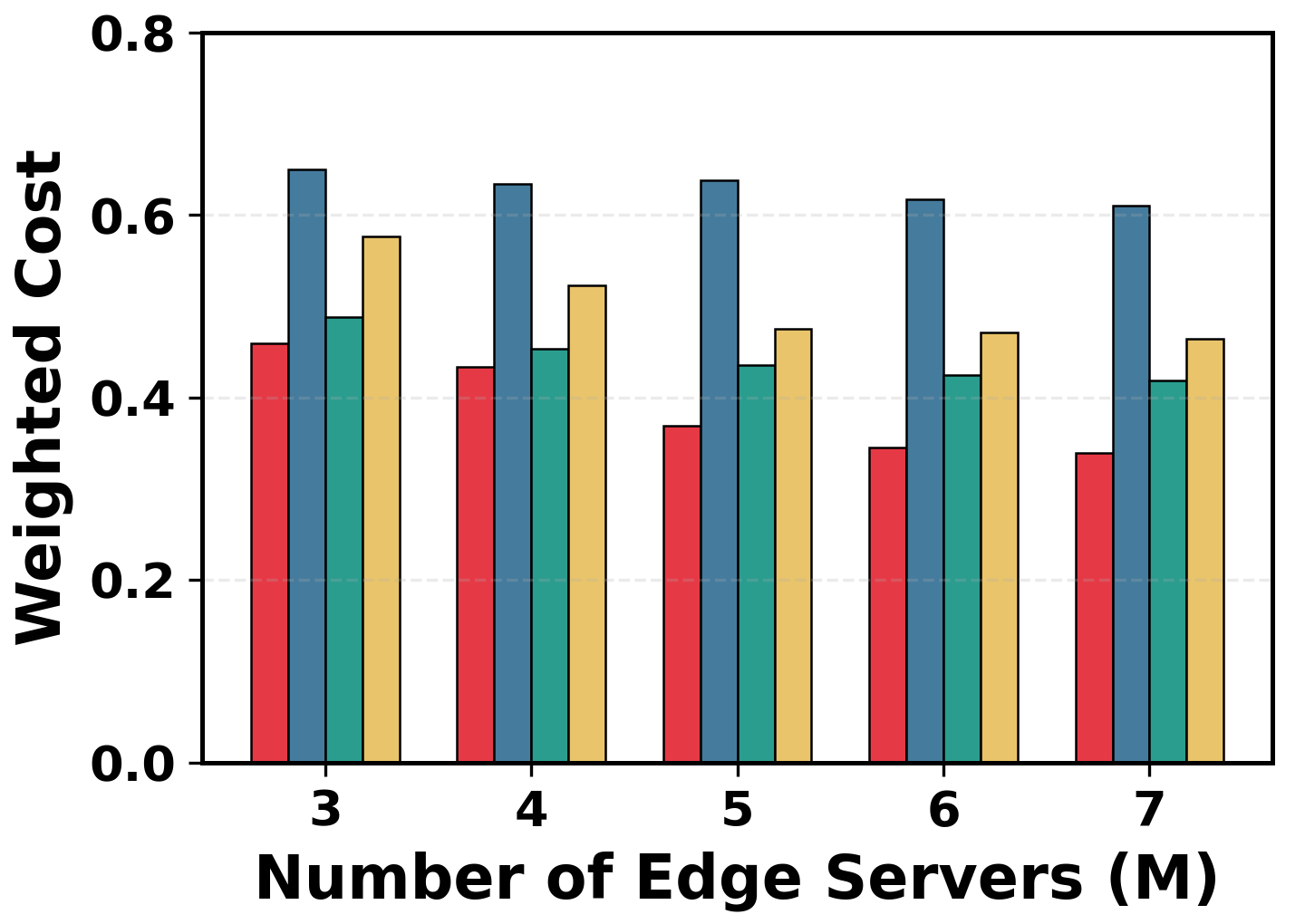}
        \caption{Varying number of ESs}
        \label{fig:vary_m}
    \end{subfigure}
    \caption{Impact of system scale on performance.}
    \label{fig:scalability}
\end{figure}

Fig.~\ref{fig:scalability} examines the impact of the number of vehicles $N$ and ESs $M$ on weighted cost. As $N$ grows from 10 to 50, CRAFT's cost falls from 0.41 to 0.34 (with a slight rise at $N=30$ from intensified resource contention), as a larger pool of participants strengthens the generalization of the federated global model. Increasing $M$ from 3 to 7 yields an even sharper drop (0.46 to 0.34, $\sim$26\%), since extra ESs add edge capacity and expand each vehicle's offloading targets, letting the constraint-aware policy find better latency-energy trade-offs. LF-MAPPO stays insensitive to scale in both cases (0.61--0.71), confirming its energy-optimization bottleneck, while Fed-MADDPG and DPTORA improve but less than CRAFT. CRAFT thus attains the lowest cost in every configuration, with its margin widening as the system scales.

\subsection{Security Analysis}

\begin{figure}[t]
    \centering
    \includegraphics[width=0.85\columnwidth, height=4cm]{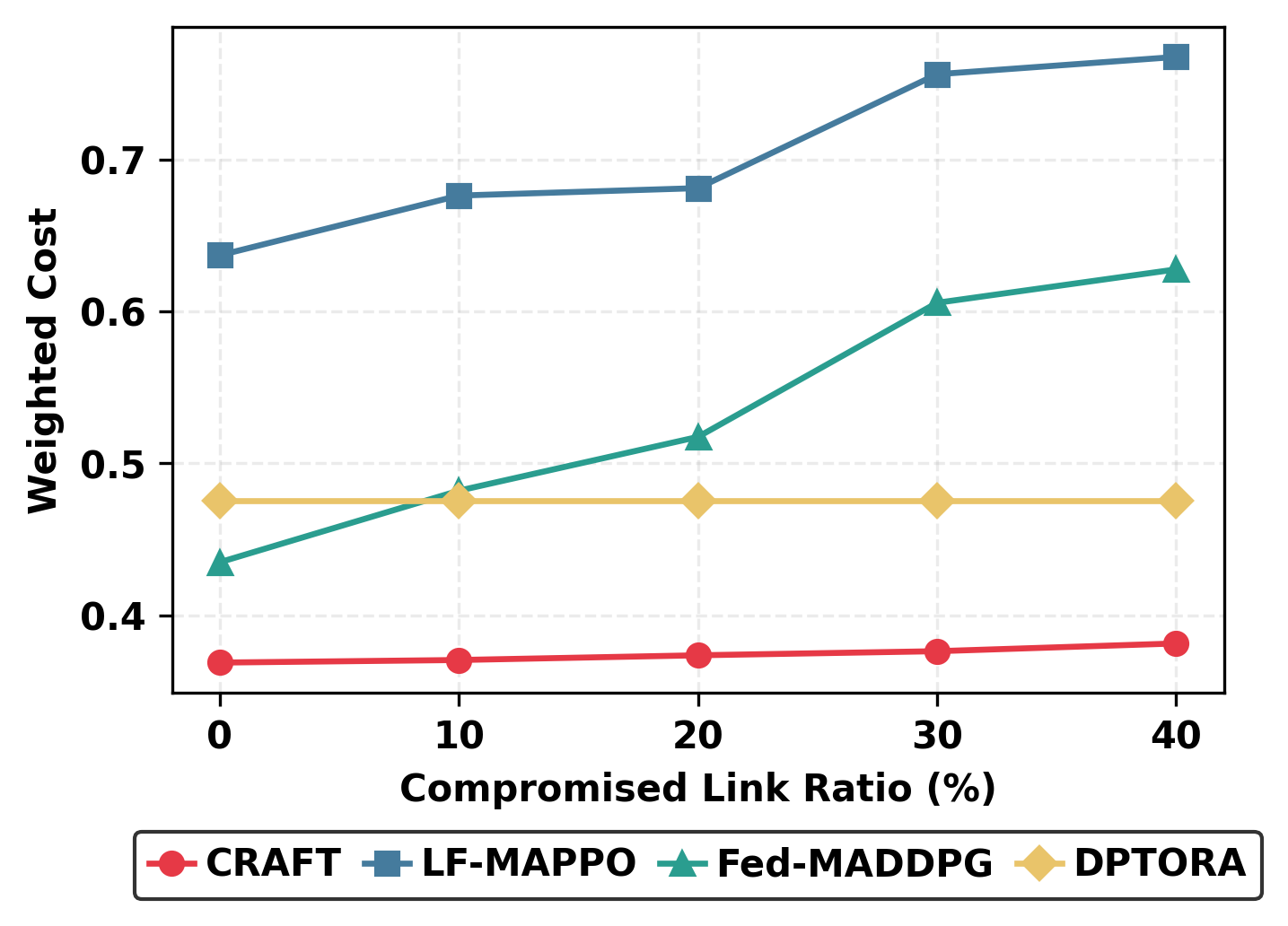}
    \caption{Impact of compromised link ratio on performance.}
    \label{fig:security}
\end{figure}

\textbf{Threat Model.} We consider an adversarial setting where a fraction of the V2I links are compromised, allowing attackers to intercept and tamper with the transmitted parameters; this is especially relevant to VEC, whose broadcast wireless channels are inherently susceptible to man-in-the-middle attacks. While an attacker may inject an arbitrary payload $\mathbf{z}'$, we adopt two representative attacks, each compromised link independently applying one with equal probability: (i) Gaussian noise injection, $\mathbf{z} \!\to\! \mathbf{z} + \boldsymbol{\delta}$ with $\boldsymbol{\delta} \sim \mathcal{N}(\mathbf{0}, \sigma_{\mathrm{atk}}^2 \mathbf{I})$; and (ii) sign-flipping, $\mathbf{z} \!\to\! -\mathbf{z}$. CRAFT's integrity protection holds against arbitrary tampering, not just these two. We evaluate integrity only, as confidentiality follows from the provable security of CKKS: without the decryption key, no entity can recover the parameters from the ciphertext~\cite{cheon2017ckks}.

As the compromised ratio rises from 0\% to 40\% (Fig.~\ref{fig:security}),  CRAFT degrades only marginally, since tampered messages fail ML-DSA verification and are discarded before aggregation. By contrast, Fed-MADDPG degrades from 0.44 to 0.63 ($\sim$43\%) and LF-MAPPO from 0.64 to 0.77 ($\sim$20\%). DPTORA is naturally immune for lack of inter-agent parameter sharing, but its baseline cost remains uncompetitive.

\subsection{Security Overhead Analysis}

\begin{figure}[t]
    \centering
    \begin{subfigure}[t]{0.48\columnwidth}
        \includegraphics[width=\textwidth, height=3cm]{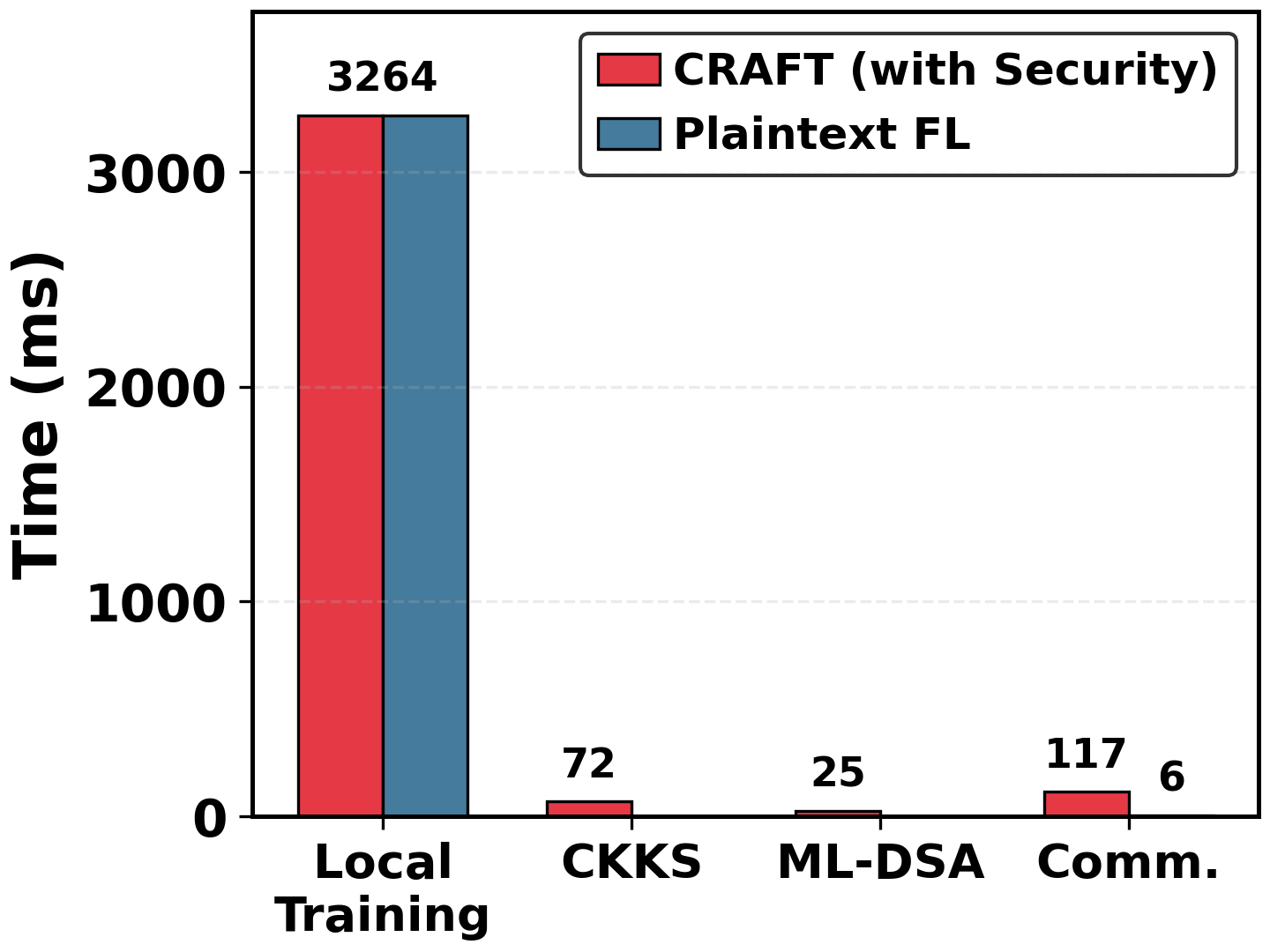}
        \caption{Per-round breakdown}
        \label{fig:fhe_breakdown}
    \end{subfigure}
    \hfill
    \begin{subfigure}[t]{0.48\columnwidth}
        \includegraphics[width=\textwidth, height=3cm]{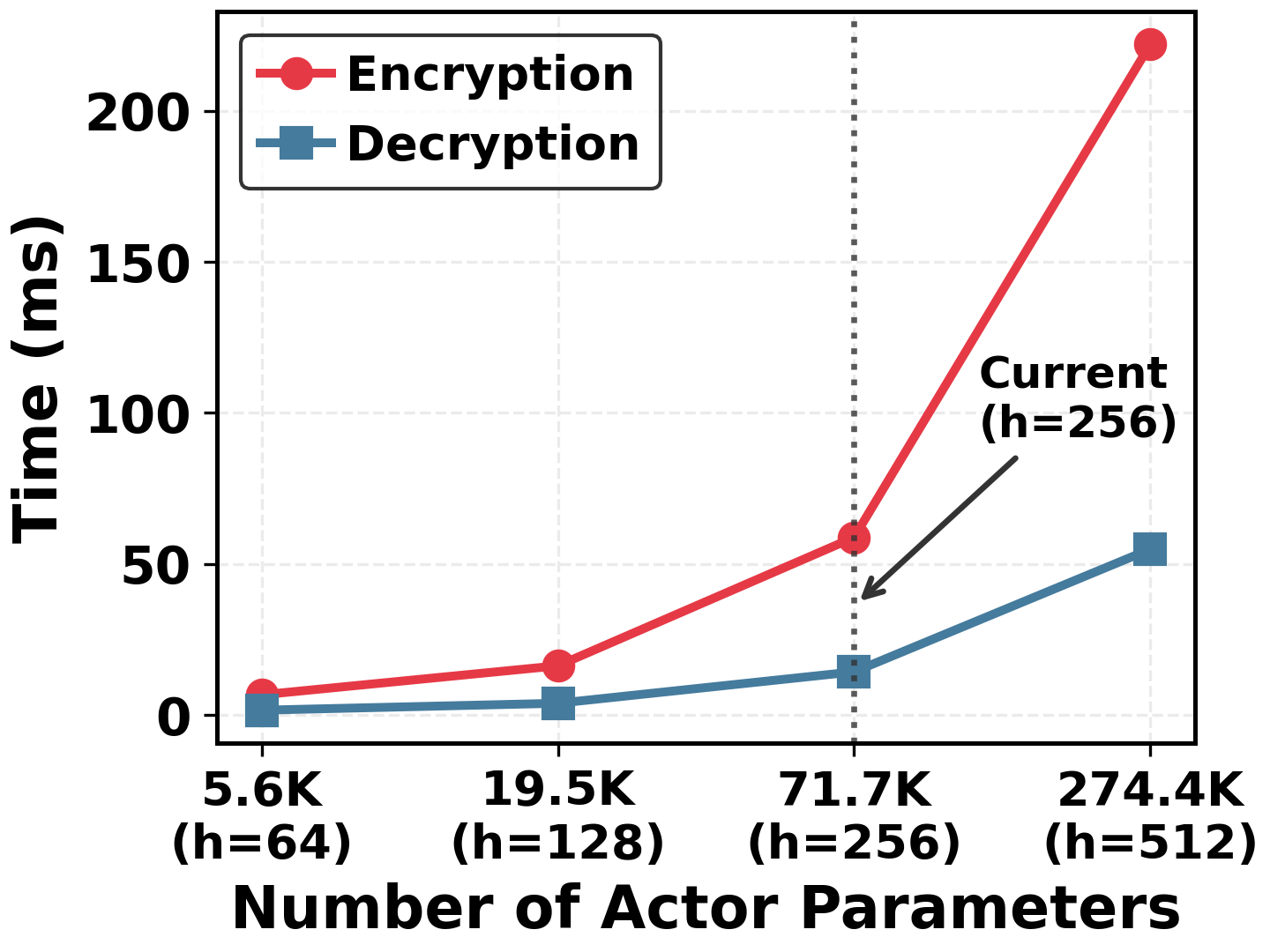}
        \caption{Scaling with model size}
        \label{fig:fhe_scaling}
    \end{subfigure}
    \caption{Security overhead analysis.}
    \label{fig:fhe}
\end{figure}

Fig.~\ref{fig:fhe} quantifies the security overhead. In the per-round breakdown (Fig.~\ref{fig:fhe_breakdown}), local training time is identical with and without security (3264~ms). CKKS encryption and decryption add 72~ms, ciphertext expansion raises communication time from 6~ms to 117~ms, and ML-DSA signing and verification, accumulated over the $K=20$ per-slot exchanges and the aggregation round, contribute about 25~ms. The total overhead of 208~ms accounts for 6.4\% of the training time, indicating a modest amortized security cost relative to local training. Fig.~\ref{fig:fhe_scaling} shows the scaling of the dominant component, CKKS. At $h=256$ (71.7K parameters), encryption and decryption take 58~ms and 14~ms; even at $h=512$ (274.4K parameters) they remain practical at about 220~ms and 55~ms. ML-DSA overhead is independent of model size, operating on a fixed-size hash digest, and since FL aggregation runs only periodically, the CKKS and communication costs are further amortized and never become a bottleneck.

\section{Conclusions and Future Work}
\label{sec:conclusion}
This paper proposes CRAFT for joint offloading and resource allocation of DAG-structured tasks in VEC. CRAFT models the problem as a CMDP solved through four cooperating components: CTDE-based multi-agent decision-making, Lagrangian relaxation for constraint handling, two-tier federated aggregation with dual-role proximal regularization, and post-quantum security combining CKKS encryption for confidentiality with ML-DSA signatures for integrity. Experiments on a VM-based platform with a realistic Melbourne CBD network show that CRAFT lowers weighted cost by about 16\%, 20\%, and 42\% over three state-of-the-art baselines, holds the completion rate (96\%) near the constraint threshold, and incurs only 6.4\% security overhead.

Future work proceeds along three directions: (i) incorporating Vehicle-to-Vehicle (V2V) cooperative offloading to exploit idle resources on nearby vehicles via joint V2I-V2V scheduling; (ii) adaptive multi-objective optimization that adjusts the latency-energy weights to real-time traffic and application demands; and (iii) attaching zero-knowledge proofs of well-formedness to CKKS-encrypted updates, so that maliciously tampered parameters from compromised-but-certified vehicles can be filtered in the ciphertext domain.

\ifCLASSOPTIONcaptionsoff
  \newpage
\fi

\bibliographystyle{IEEEtran}
\bibliography{IEEEabrv,Bibliography}
\vspace{-6mm}

\vfill

\end{document}


\onecolumn

\appendix

\begin{table}[h]
\centering
\caption{Summary of key notations.}
\label{tab:notation}
\renewcommand{\arraystretch}{1.15}
\begin{tabular}{cl||cl}
\hline
\textbf{Notation} & \textbf{Description} & \textbf{Notation} & \textbf{Description} \\
\hline
$\mathcal{T}$, $T$, $\tau$ & Set/number of time slots; slot duration & $\Lambda_n$ & Task arrival rate of vehicle $n$ \\
$\mathcal{N}$, $N$ & Set and number of vehicles & $G_n\!=\!(\mathcal{V}_n,\mathcal{E}_n)$ & DAG of vehicle $n$ (node set, edge set) \\
$\mathcal{M}$, $M$ & Set and number of ESs & $\mathcal{K}_n$, $K_n$ & Set and number of subtasks \\
$R_m$ & Coverage radius of ES $m$ & $v_{\mathrm{src}}$, $v_{\mathrm{snk}}$ & Virtual source and sink nodes \\
$F_m$ & Computing capacity of ES $m$ (cycles/s) & $d_k$, $c_k$, $o_k$ & Input size, workload, output size of subtask $k$ \\
$\mathbf{p}_n(t)$ & Position of vehicle $n$ in time slot $t$ & $w_{i,k}$ & Data transfer weight of edge $(i,k)$ \\
$\mathbf{v}_n(t)$, $v_n(t)$ & Velocity vector and speed of vehicle $n$ & $a_k$ & Offloading decision for subtask $k$ \\
$\mathbf{q}_m$ & Position of ES $m$ & $f_n^{\mathrm{loc}}$, $\kappa$ & Local CPU frequency; capacitance coefficient \\
$d_{n,m}(t)$ & Distance between vehicle $n$ and ES $m$ & $f_{k,m}^{\mathrm{edge}}$ & Edge resources for subtask $k$ on ES $m$ \\
$\mathcal{M}_n(t)$ & Set of ESs reachable by vehicle $n$ & $D_k^{\mathrm{comp}}$, $E_k^{\mathrm{comp}}$ & Computation delay and energy of subtask $k$ \\
$\boldsymbol{\Delta}_{n,m}(t)$ & Relative position $\mathbf{p}_n(t) - \mathbf{q}_m$ & $D_{i,k}^{\mathrm{trans}}$, $E_{i,k}^{\mathrm{trans}}$ & Transfer delay and energy of edge $(i,k)$ \\
$D_{n,m}^{\mathrm{res}}(t)$ & Residual sojourn time of vehicle $n$ at ES $m$ & $\mathrm{pred}(k)$, $\mathrm{succ}(k)$ & Predecessor and successor sets of subtask $k$ \\
$B$, $b_{n,m}(t)$ & Total/allocated bandwidth & $D_k^{\mathrm{start}}$, $D_k^{\mathrm{finish}}$ & Start and finish times of subtask $k$ \\
$\mathcal{N}_m(t)$ & Set of vehicles connected to ES $m$ & $D_n^{\mathrm{DAG}}$ & Makespan of vehicle $n$'s DAG task \\
$h_{n,m}(t)$ & Channel gain between vehicle $n$ and ES $m$ & $E_n^{\mathrm{total}}$ & Total vehicle-side energy of vehicle $n$ \\
$h_0$, $\alpha$ & Reference channel gain; path loss exponent & $\mathbf{a}(t)$, $\mathbf{f}(t)$ & Offloading and resource allocation vectors \\
$\xi_{n,m}(t)$ & Rayleigh fading coefficient & $s_n(t)$ & Task success indicator of vehicle $n$ \\
$P_n^{\mathrm{tx}}$, $P_m^{\mathrm{ES}}$ & Transmit power of vehicle $n$ / ES $m$ & $D_n^{\max}$ & Deadline of vehicle $n$'s DAG task \\
$\sigma^2$ & Noise power & $\hat{D}_k(t)$ & Last interaction time of subtask $k$ at its ES \\
$r_{n,m}^{\mathrm{up}}$, $r_{n,m}^{\mathrm{down}}$, $r^{\mathrm{wire}}$ & Uplink, downlink, and backhaul rates & $\omega_1$, $\omega_2$, $\theta$ & Latency/energy weights; completion rate threshold \\
\hline
\end{tabular}
\end{table}


\begin{table}[h]
\centering
\caption{Experimental parameters.}
\label{tab:params}
\renewcommand{\arraystretch}{1.15}
\begin{tabular}{clc}
\hline
\textbf{Category} & \textbf{Parameter} & \textbf{Value} \\
\hline
\multirow{7}{*}{\textit{System}}
 & Road network area & ${\sim}2$ km $\times$ $2$ km \\
 & Number of ESs, $M$ & 5 (default) \\
 & Number of vehicles, $N$ & 30 (default) \\
 & Time slot duration, $\tau$ & 1 s \\
 & Vehicle speed & 30--50 km/h \\
 & Vehicle VM configuration & 4--12 vCPUs, 4--32 GB RAM \\
 & ES VM configuration & 16 vCPUs, 64 GB RAM \\
\hline
\multirow{8}{*}{\textit{Communication}}
 & System bandwidth, $B$ & 20 MHz \\
 & Vehicle transmit power, $P_n^{\mathrm{tx}}$ & 0.3 W \\
 & ES transmit power, $P_m^{\mathrm{ES}}$ & 2 W \\
 & Path loss exponent, $\alpha$ & 3 \\
 & Reference channel gain, $h_0$ & $-30$ dB \\
 & Noise power spectral density, $N_0$ & $-174$ dBm/Hz \\
 & ES coverage radius, $R_m$ & 500 m \\
 & Backhaul rate, $r^{\mathrm{wire}}$ & 1 Gbps \\
\hline
\multirow{15}{*}{\textit{Training}}
 & Learning rate (Actor/Critic) & $3\times10^{-4}$ \\
 & Embedding layers & 3-layer FC, $h\!=\!256$ \\
 & Attention heads & 1 \\
 & Attention dimension & 256 \\
 & Activation function & Tanh \\
 & Discount factor, $\gamma$ & 0.99 \\
 & GAE parameter, $\varsigma$ & 0.95 \\
 & Clipping parameter, $\epsilon$ & 0.2 \\
 & Experience batch size, $T_{\mathrm{batch}}$ & 64 \\
 & FedProx proximal coefficient, $\mu$ & 0.1 \\
 & FL aggregation period, $K$ & 20 \\
 & Completion rate threshold, $\theta$ & 0.95 \\
 & Lagrangian multiplier LR, $\eta_\lambda$ & 0.01 \\
 & Lagrangian multiplier upper bound, $\lambda^{\max}$ & 10 \\
 & Weighting coefficients, $\omega_1$, $\omega_2$ & 0.5, 0.5 \\
\hline
\multirow{6}{*}{\textit{Security}}
 & FHE scheme / Library & CKKS / TenSEAL \\
 & Polynomial ring dimension, $N_{\mathrm{poly}}$ & 8192 \\
 & Scaling factor & $2^{40}$ \\
 & Signature scheme / Library & ML-DSA-44 / liboqs \\
 & Signature size & 2,420 bytes \\
 & Attack noise variance, $\sigma_{\mathrm{atk}}^2$ & 0.01 \\
\hline
\end{tabular}%
\end{table}